\documentclass[11pt,letterpaper]{article}
\usepackage[utf8]{inputenc}
\usepackage{fullpage}

\title{Online Virtual Machine Allocation with Predictions} 

\author{Niv Buchbinder\thanks{Tel Aviv University, niv.buchbinder@gmail.com}\and 
Yaron Fairstein\thanks{Technion, yyfairstein@gmail.com} \and 
Konstantina Mellou\thanks{Microsoft Research, kmellou@microsoft.com} \and
Ishai Menache\thanks{Microsoft Research, ishai@microsoft.com}\and 
Joseph (Seffi) Naor \thanks{Technion, naor@cs.technion.ac.il}}

\usepackage{amsmath}
\usepackage{amsthm}
\usepackage{amssymb}
\usepackage{amsfonts}
\usepackage{graphics}
\usepackage{color}
\usepackage{ dsfont }
\usepackage[procnumbered, linesnumbered,algoruled]{algorithm2e}

\newtheorem{thm}{Theorem}[section]

\newtheorem{lem}[thm]{Lemma}
\newtheorem{lemma}[thm]{Lemma}
\newtheorem{obs}[thm]{Observation}
\newtheorem{corollary}[thm]{Corollary}

\newtheorem{defn}[thm]{Definition}

\def \OPT  {\mbox{\rm OPT}}

\def \II   {{\mathcal I}}
\def \uni {{g}}
\def \HH {{\Pi}}
\def \OPT {{OPT}}
\def \OA {{\mbox{OPT}_{avg}}}
\def \OPTS {{OPT^S}}
\date{}

\begin{document}

\maketitle

\begin{abstract}
    The cloud computing industry has grown rapidly over the last decade, and with this growth there is a significant increase in demand for compute resources. Demand is manifested in the form of Virtual Machine (VM) requests, which need to be assigned to physical machines in a way that minimizes resource fragmentation and efficiently utilizes the available machines. This problem can be modeled as a dynamic version of the bin packing problem with the objective of minimizing the total usage time of the bins (physical machines). Earlier works on dynamic bin packing assumed that no knowledge is available to the scheduler and later works studied models in which lifetime/duration of each ``item'' (VM in our context) is available to the scheduler. This extra information was shown to improve exponentially the achievable competitive ratio.
    
    Motivated by advances in Machine Learning that provide good estimates of workload characteristics,
    this paper studies the effect of having extra information regarding future (total) demand. In the cloud context, since demand is an aggregate over many VM requests, it can be predicted with high accuracy (e.g., using historical data). We show that the competitive factor can be dramatically improved by using this additional information; in some cases, we achieve constant competitiveness, or even a competitive factor that approaches $1$. Along the way, we design new offline algorithms with improved approximation ratios for the dynamic bin-packing problem. 
\end{abstract}

\section{Introduction}\label{subsec:motivation}

Cloud computing is a growing business which has revolutionized the way computing resources are consumed. The emergence of cloud computing is attributed to lowering the risks for end-users (e.g., scaling-out resource usage based on demand), while allowing providers to reduce their costs by efficient management and operation at scale. One popular way of consuming cloud resources is through Virtual Machine (VM) offerings. Users rent VMs on demand with the expectation of a seamless experience until they decide to terminate usage. In turn, cloud resource managers place VMs on physical servers that have enough capacity to serve them. The specific VM allocation decisions have a direct impact on resource efficiency and return on investment. For example, inefficient placement mechanisms might result in fragmentation and unnecessary over-provisioning of physical resources. 

Our goal in this paper is to design algorithms for allocating VMs to physical machines in a cloud facility (e.g., cluster, region), so that the total active machine-time, taken over all machines, is minimized; a machine is considered {\em active} if one or more VMs run on it. When a machine becomes inactive, it can be returned to the general pool of machines, and therefore does not contribute to the cost function. In certain scenarios, the same optimization can also lead to power savings, under the assumption that empty machines can be kept in an idle, low-power mode \cite{guenter2011managing,DynamicRightSizing}. 

The problem of allocating VMs to physical machines can be modeled as a generalization of the classic (and extensively studied) {\em static bin packing} problem, where the goal is to pack a set of items of varying sizes, while minimizing the number of bins used \cite{johnson1973near,man1996approximation,coffman2007performance}. The VM allocation problem corresponds to a {\em dynamic bin packing} problem in which items, or VMs, arrive over time and later depart \cite{coffman1983dynamic}. Minimizing the total active-machine time translates then to minimizing the total usage time of the bins, or machines \cite{YXW14,YXW16,RX16,azar2017tight}. The VM allocation problem is of interest in both the {\em uniform size} case, in which all items have the same size (and each bin can pack at most $\uni$ items)  \cite{flammini2010minimizing}, but especially under the more general setting, which we refer to as the {\em non-uniform size} case. The problem is known to be NP-hard even in the uniform case when $\uni=2$~\cite{winkler2003wavelength}. 

The dynamic bin packing problem has been studied in both offline and online settings. In the online setting, items arrive over time, giving rise to two different models. In the {\em non-clairvoyant} case \cite{flammini2010minimizing,YXW14,tang2016first} no information is given to the scheduler upon arrival of a new item, and indeed only poor performance is obtained when there is a large variation in item duration times \cite{KL15} (see additional discussion later). In the {\em clairvoyant} setting \cite{YXW16,RX16,azar2017tight} the departure time (or duration) of an item is revealed upon arrival, allowing for significant performance improvements. 

The clairvoyant model assumes that highly accurate lifetime predictions are available to a scheduler. In the cloud context, this information has recently been obtained through Machine Learning (ML) tools \cite{ResourceCentral,lifeTimePatent,li2015play}, which are deployed to support resource management decisions for the underlying systems (see \cite{li2017deep, boutaba2018comprehensive, ResourceCentral, gao2014machine} and references therein). ML is increasingly used, not only for lifetime prediction, but also to predict other metrics, such as machine health \cite{ResourceCentral} and future demand \cite{guenter2011managing}. 

Motivated by the recent momentum in applying ML for cloud systems, we take the online dynamic bin packing model a step further, and study the advantage of having \emph{additional} information, on top of VM or item lifetimes. Specifically, we focus on designing online algorithms that possess some form of prediction about future \emph{demand}. From a practical perspective, we note that demand is an \emph{aggregate} over numerous requests; as such, it can be predicted with high accuracy \cite{guenter2011managing,zhang2016history} (in fact, higher accuracy than individual VM lifetime predictions).

\subsection{Our Results}

We first describe the setting in which we study the VM scheduling problem. We assume that each VM (item) has a demand (size), and each machine (bin) has unit size. Thus, the total demand of VMs assigned to a machine at any point of time cannot exceed $1$. In the uniform size case we assume that the size of all VMs is $\frac{1}{\uni}$ for some integer $\uni$. The VMs arrive over time and need to be assigned to machines for their duration (lifetime) in the system. As there is no migration of VMs across physical machines, the initial assignment remains as it is until the VM terminates. 

We can assume without loss of generality that at any given time there is at least one active VM. We refer to the aggregate size of the VMs that are active at time $t$ as the total demand/load at $t$. A physical machine is considered active when one or more VMs are assigned to it. The goal is to minimize the total time that the machines remain active. An important parameter in our results is $\mu$, defined to be the ratio between the maximum and minimum duration, taken over all VMs.  Finally, let $\HH_k$ be the asymptotic competitive ratio of the harmonic bin packing algorithm\footnote{The harmonic algorithm is parameterized by $k$, which controls an additive term in its competitive ratio. $\HH_k$ is a monotonically decreasing number that approaches  $\HH_\infty\approx 1.691$. $\HH_k$ quickly becomes close to $1.691$, for example, $\HH_6 = 1.7$ and $\HH_{12}\approx 1.692$.} \cite{lee1985simple}. 

As there is always at least one active VM in each time step, the optimal cost of the dynamic bin packing problem is at least $T$, the length of the time horizon. To facilitate the understanding of our results, we divide both the optimal cost, $\OPT$, as well as our algorithm's cost by $T$. Let $\OA$ denote the optimal cost divided by $T$. The value $\OA$ should thus be read as the average number of machines used by an optimal solution. This change, of course, does not affect the multiplicative factor in the approximation/competitive ratios we get. However, any additive term should now be read as the average number of additional machines our algorithm is using over the (average) number of machines an optimal solution is using. We note that in our cloud computing context the average number of machines is typically in the order of thousands.

We study the VM scheduling problem in both offline and online settings. 

\subsubsection{Offline Algorithms}

We first show the following result for the offline problem.

\begin{thm}\label{thm:main-offline}
For any $k\geq 3$, there exist {\bf offline} scheduling algorithms whose {\bf average} cost is at most:
\begin{center}
{\renewcommand{\arraystretch}{1.2}
\begin{tabular}{l|c|c}
& {\rm Non-uniform size case} & {\rm Uniform size case}  \\ \hline
{\rm Algorithms \ref{DensityAlgorithm},  \ref{HarmonizeOffline2}} & $\HH_k\cdot \OA +  O\left( \sqrt{\OA \cdot k\cdot \log \mu}\right)$ & $\OA + O\left( \sqrt{ \OA \cdot\log \mu}\right)$\\
{\rm Algorithms \ref{CoveringAlgorithm},  \ref{HarmonizeOffline2}} &  $2\HH_k (1+\frac{1}{k-2})\cdot \OA + k$ & $2\cdot \OA$\\ \hline
{\rm Previous results} &  $4\cdot \OA$~\cite{RX16} & $2\cdot \OA$~\cite{alicherry2003line,KR05}
\end{tabular}}
\end{center}
\end{thm}

The $2\cdot \OA$ upper bound for the uniform size case is well known \cite{alicherry2003line,KR05,RX16}. However, it is described here not only to compare against our other results, but also because the techniques used to prove it are later used in the online case; though fairly simple, these techniques are somewhat different from previous proofs. In the above theorem we obtain two improved new bounds for the offline case. These improvements are in the spirit of the {\em asymptotic} approximation ratio commonly used in the standard bin packing problem. Algorithm \ref{CoveringAlgorithm}  has multiplicative approximation ratio $2\HH_k(1+\frac{1}{k-2})$, while using extra $k$ machines in each time step. Thus, its ``asymptotic" approximation approaches $2\HH_{\infty} \approx 3.38$, which is better than the best known (strict) $4$-approximation for the problem\footnote{We remark that Algorithm \ref{CoveringAlgorithm} also achieves the same $4$-approximation (with no additive term).}. 

Algorithm \ref{DensityAlgorithm} achieves an even better multiplicative approximation ratio of $\Pi_k$ (that approaches $1.69$), albeit only when the average number of machines used by an optimal solution is relatively large. Specifically, it uses an extra $ O(\sqrt{\OA\cdot k \cdot\log \mu})$ machines in each time step. If the average number of machines used in the optimal solution is much larger than $\log \mu$, the latter additive term becomes negligible compared to $\OA$. In the uniform size case, the asymptotic approximation of the algorithm is $1$, as may be expected when sizes are uniform. When the maximum demand of any VM is small (and also in some other scenarios), our performance guarantees are better than those outlined in Theorem \ref{thm:main-offline} (actually, in both offline and online settings). We refer the reader to Section \ref{sec:offline} and Section \ref{full_predictions} for more details.

\subsubsection{Online Algorithms with Additional Information}

Our main contribution in this paper is the construction of new online algorithms that have access to additional information on future demand, leading to improved competitive ratios. Interestingly, our online algorithms are inspired by their offline counterparts. Earlier results assumed that an online scheduler gets no extra information upon arrival of a VM request. The performance guarantee of these algorithms turned out to be very poor in many cases. Better results were later obtained for the clairvoyant model, in which duration of requests are revealed upon arrival. 

Specifically, we explore two novel models in which the scheduler is provided with predictions about demand. In the first model, the average load is known to the scheduler (a single value), and in the second model the total load in each of the future time steps is known to the scheduler. We remark that predicting future cumulative demand is much simpler than obtaining the full structure of an instance, which requires predicting future arrivals of individual requests.

\begin{thm}\label{thm:main-online}
For any $k\geq 2$, there exist {\bf online} scheduling algorithms with {\bf average} cost at most:
\begin{center}
{\renewcommand{\arraystretch}{1.2}
\begin{tabular}{l|c|c}
{\rm Extra information} & {\rm Non-uniform size case} & {\rm Uniform size case}  \\ 
{\rm beyond duration} & \\
\hline
{\rm Average load} & $\HH_k\cdot \OA + k\cdot O(\sqrt{\OA\cdot \log \mu})$ & $\OA + O(\sqrt{\OA\cdot \log \mu})$\\
{\rm Future load vector} &  $8\cdot \OA$ & $2\cdot \OA$\\ \hline
{\rm No extra information} & $O(\sqrt{\log \mu})\cdot \OA$~~~{\rm \cite{azar2017tight}} & $O(\sqrt{\log \mu})\cdot \OA$~~~{\rm \cite{azar2017tight}}
\end{tabular}}
\end{center}
\end{thm}
In the above table, we compare our results with the previously best known online result in the clairvoyant model, due to Azar and Vainstein \cite{azar2017tight}\footnote{We note that similarly to the algorithm of \cite{azar2017tight}, our algorithm also does not need to know the value of $\mu$ upfront.}. As indicated in the table, \cite{azar2017tight} designed an algorithm whose total cost is at most $O(\sqrt{\log \mu})\cdot \OA$, and proved that this ratio is optimal. Our results demonstrate that with more information the competitive ratio can be dramatically improved. Suppose that the only additional information provided is the average load (taken over the full time horizon), and that the average number of machines used is much larger than $\log \mu$; then, we obtain a constant competitive ratio that approaches $\HH_{\infty}\approx 1.69$ in the non-uniform size case and an asymptotic ratio of  $1$ in the uniform size case. Thus, our performance guarantee is always better than \cite{azar2017tight}, and it is the same when the average load is $O(1)$. 

If the load at all future times is known, we achieve a (strict) constant competitive ratio under no additional assumptions. This is in contrast to the $\Omega(\sqrt{\log \mu})$ lower bound on the competitive ratio of any algorithm without this extra knowledge~\cite{azar2017tight}.

In Section \ref{sec:inaccurate} we complement our results and analyze the performance of our algorithms when the average load prediction, as well as interval lengths predictions, are inaccurate. 

We complement the above results by generalizing the lower bound of \cite{azar2017tight}  to take into account also $\OA$ showing that the additive term $O(\sqrt{\OA \cdot \log \mu})$ is indeed unavoidable, if only the average future load (and lifetime) is available to an online algorithm.  

\begin{thm}\label{thm:lower-bound11}
The average cost of any online algorithm is at least $\Omega\left(\sqrt{\OA\cdot \log \mu}\right)$. The bound holds even for the uniform size case and with prior knowledge of the average load and $\mu$. \end{thm}

We also remark that  the lower bound of approximately $1.542$~\cite{balogh2019new} on the asymptotic competitive ratio of any static online bin packing algorithm, carries over to the dynamic clairvoyant case\footnote{Simply use the same (long) duration for all requests that arrive (almost) at the same time according to the adversarial arrival sequence.}. 

\subsection{Techniques} 

The main issue we cope with in the online setting is how to improve the competitive ratio by utilizing the additional information provided to an online scheduler. At a high level, this is achieved by drawing on ideas from our new offline algorithms; we show that these algorithms can be to some extent ``simulated" in the online case, even when less information is available. Yet, the loss to performance is bounded.

\paragraph{How to utilize future load predictions?}
When loads for each future time step are available, we draw on ideas from Algorithm \ref{CoveringAlgorithm}. In each iteration, this offline algorithm considers the unscheduled requests, and finds greedily (and carefully) a set of requests (among the unscheduled requests) that can be scheduled on one or two machines, and having high enough load in every time step. This set can be interpreted as a cover of the time horizon. The offline algorithm then repeats this process with the remaining unscheduled requests, till all requests are scheduled.

Achieving this goal online is tricky, as multiple covers of the time horizon must be created in parallel without knowing future requests. Each ``error" in assigning requests to covers may either increase the number of machines required for scheduling a cover, or increase the number of covers (again, resulting in too many active machines). We show that when given information on future demand, the number of ``extra" covers generated is bounded. The high level idea is to maintain several open machines (and not just two as in the offline case), and schedule a new request on the lowest index machine that ``must" accept it in order to preserve a ``high load" invariant. Finding the right machine is done online utilizing the predictions on the remaining future demand and the new request's lifetime. Surprisingly, we are able to get a constant competitive factor in this case, even though the online scheduler is not familiar with the full interval structure of the instance, as in the offline setting, only with cumulative load. The results are presented in Section \ref{full_predictions}.

\paragraph{How to use average load prediction?}
Interestingly, we show that this {\em single value} parameter can be extremely useful in improving the competitive factor. This is done by mimicking Algorithm~\ref{DensityAlgorithm}. This offline algorithm finds a dense subset of requests of roughly the same duration that can be scheduled together (similarly in spirit to the greedy set-cover algorithm). In the offline setting we show that this is possible whenever there exists a point in time in which demand is high enough. 

The online scheduler is not familiar with the demand ahead of time. Hence, it uses a careful classification of the requests by their duration, the current demand, and the average demand. While the idea of classification has been used before in the context of dynamic bin packing, we classify intervals in a more sophisticated way. Finally, to get our refined bounds we schedule each class of intervals using a new family of non-clairvoyant algorithms (discussed in the sequel) that trade off carefully multiplicative and additive terms. The results are presented in Section \ref{online_average}.

\paragraph{A new family of non-clairvoyant algorithms.}
To get our refined bounds we show a general reduction that transforms any $k$-bounded space (static) bin packing algorithm (see exact definitions in Appendix \ref{app:staticAlgs}) into a non-clairvoyant algorithm for the dynamic bin packing problem. Note that the optimal cost in the static bin packing is simply the number of bins (and not the total duration). Hence, we use $\OPTS$ to emphasize that this is the optimal solution for static instances. We prove the following.

\begin{lemma}\label{lem-reduction1}
Given a $k$-bounded space bin packing algorithm whose cost at most $c \cdot \OPTS +\ell$, there exists an online non-clairvoyant algorithm for the dynamic bin packing setting whose {\bf average} cost is at most
$c \cdot \mu\cdot  \OA + \max\{k,\ell\}$.
\end{lemma}

For example, substituting in this theorem the performance of the Harmonic Algorithm, we obtain a non-clairvoyant algorithm whose average cost is at most $\HH_k \cdot \mu\cdot  \OA + k$.

\subsection{Related Work}\label{sec:related}
In the remainder of this section, we discuss additional relevant work from an algorithmic perspective. Flammini et al.~\cite{flammini2010minimizing} analyzed a natural First-Fit heuristic for the offline problem, and proved it is a $4$-approximation. Tang et al.~\cite{tang2016first} proved that a First-Fit heuristic in the online non-clairvoyant setting is $\mu+4$ competitive. This was proven to be almost optimal by Li et al.~\cite{YXW16} who showed that the competitiveness of any Fit-Packing algorithm cannot be better than $(\mu+1)$. Ren et al.~\cite{RX16} designed a First-Fit based algorithm for the clairvoyant bin packing problem. Using additional predictions of maximum and minimum lifetimes of items they achieved competitive ratio $2\sqrt{\mu}+3$.

The online dynamic bin packing problem was first introduced by Coffman et al.~\cite{coffman1983dynamic}. Their objective was minimizing the maximum number of active machines over the time horizon. They designed an $2.788$-competitive algorithm. For this model, Wong et al.~\cite{wong20128} obtain a lower bound of $\frac{8}{3}\approx 2.666$ on the competitive ratio.

Additional interval scheduling models have been recently considered in the context of cloud computing (e.g., \cite{lucier2013efficient,azar2015truthful,chawla2017truth,Khuller18} and references therein). These models are fundamentally different than our VM scheduling setup, mainly because the ``jobs'' in these works have some flexibility (termed slackness) as to when they are executed. Finally, we note that there has been growing interest in designing resource management algorithms with ML-assisted (and potentially inaccurate) predictions; see, e.g., recent work on online caching \cite{lykouris2018competitive}, scheduling \cite{lattanzi2020online} and the ski-rental problem \cite{purohit2018improving}. 

\vspace{8pt}

\noindent\textbf{Organization.} Our model is formally defined in Section \ref{sec:pre}. In Section \ref{sec:offline}, we consider the offline case. The online case is studied under two different settings: in Section \ref{online_average}, we assume that the average load information is available, whereas in Section \ref{full_predictions} the scheduler is equipped with the future load vector predictions.

\section{Model and Preliminaries}\label{sec:pre}

We model each VM request as a time \emph{interval} $I=[s,e)$; we often use the term interval when referring to a VM request. Each interval is associated with a {\em start time}, $s_I$, {\em end time}, $e_I$, and a {\em size}, $w_I\leq 1$. The intervals are scheduled on machines/bins whose size is normalized to $1$. We say that $t\in I$ if $s_I \leq t < e_I$. Let $\ell_I=e_I-s_I $ be the {\em length} of interval $I$. We assume without loss of generality that the minimum length of an interval is $1$, and denote by $\mu$ the maximum length of an interval (which is not necessarily known in advance). Let $\beta=\max_{I}w_I$ be the maximum size of an interval (which, again, is not necessarily known in advance). In the {\em uniform size model} the size of all intervals is $\frac{1}{\uni}$ for some integer value $\uni$. In the {\em non-uniform size model} the size of each interval is arbitrary.

The {\em static bin packing} problem is a classic NP-hard problem in which the goal is to pack a set of items of varying sizes, while minimizing the number of bins used. The problem has been studied extensively in both offline and online settings  \cite{johnson1973near,man1996approximation,coffman2007performance}. The problem of allocating VMs to physical machines is equivalent to the {\em dynamic bin packing} problem in which items (VMs) arrive over time and later depart \cite{coffman1983dynamic}. The goal is to minimize the total usage time of the bins which is the same as minimizing the total time the machines are active \cite{YXW14,YXW16,RX16,azar2017tight}. In the online setting items arrive over time, and in the {\em non-clairvoyant} case no information is given to the scheduler upon arrival of a new item, while in the {\em clairvoyant} setting the departure time (or duration) of an item is revealed upon its arrival.  

A machine is said to be {\em active} or {\em open} at time $t$ if at least one VM is  running on it. Our goal is to schedule the VMs so as to minimize the total (or equivalently the average) number of active machines over the time horizon. We assume that the cloud capacity is large enough, so that VM requests can always be accommodated. Without loss of generality, we further assume that at each time $t$ there is at least one active request (otherwise, the time horizon can be partitioned into separate time horizons).

Let $\II$ be a set of all intervals (VM requests). We define $\II(t)=\{I\in \II  |  t\in I\}$ as the set of intervals that are active at time $t$, and let $N_t= |\II(t)|$. Let $v=(v_1, v_2, \ldots, v_T)$ be the {\em load vector} over time, where $v_t= \lceil \sum_{I \in \II(t)}w_I \rceil$. In our analysis, we use several norms of the load vector: $\|v\|_1= \sum_{t=1}^{T} v_t$, $\|v\|_{\infty}= \max_{t=1}^{T}\{ v_t\}$, and $\|v\|_{0}= \sum_{t=1}^{T}{\bf \mathds{1}}_{(N_t>0)}$  (i.e., the total number of time epochs in which there is at least one active VM request). In addition, let $v_{avg}=\frac{\|v\|_1}{T}$ be the average value of the load vector (or the average demand). Throughout the paper, we will use load vector notions not only for the set $\II$ of all intervals, but also for different subsets $S\subseteq \II$. In every such use case, we describe explicitly the corresponding subset. A simple (known) lower bound on the value of the optimal solution, using our load vector notation, is the following:

\begin{obs} \label{CLB}
The total active-machine time required by any scheduler is at least $\|v\|_1$. 
\end{obs}

We next provide a useful lemma about intersecting intervals (intervals that are all active at the same time $t$). We use this lemma frequently in our algorithms' analyses to guarantee that in each such set, there is one interval that sees a high load, i.e., the total load is above a certain threshold, for its whole duration. 

\begin{lemma}[Intersecting intervals] \label{intersecting-intervals}
Let $\II$ be a set of intervals with load vector $v$ that are all using time $t$ (i.e.,  $t\in I$ for all $I\in \II$). If $v_t>\alpha$, then there exists an interval $I\in \II$ such that $v_{t'}>\alpha/2$ for all $t'\in I$. 
\end{lemma}
\begin{proof}
Since all intervals in $\II$ are active at time $t$, it can be seen that their load vector is non-decreasing until time $t$, and non-increasing after time $t$. Let $I_1, I_2, \ldots, I_J \in \II$ be the intervals sorted by their starting times (which are all prior to $t$). Let $A_1=\{I_1, \ldots, I_j\} \subseteq \II$  be such that $\sum_{i=1}^{j}w_{I_i} \leq  \alpha/2$, but $\sum_{i=1}^{j+1}w_{I_i} > \alpha/2$. For each interval in $\II \setminus A_1$ the load at its starting point is strictly more than $\alpha/2$. Similarly define $A_2=\{I_k,...,I_J\}$ such that $\sum_{i=k}^{J}w_{I_i} \leq  \alpha/2$, but $\sum_{i=k-1}^{J}w_{I_i} > \alpha/2$. and let $\II \setminus A_2$ be the subset of intervals whose load in their endpoint is strictly more than $\alpha/2$. As the total load in $\II$ is strictly more than $\alpha$, and the load of intervals in $A_1\cup A_2$ is at most $\alpha$, there must be an interval $I \in \II\setminus (A_1 \cup A_2)$. The load that an interval  $I \in \II\setminus (A_1 \cup A_2)$ observes is strictly more than $\alpha/2$ at both its start and end times, and hence it is strictly more than $\alpha/2$ at any $t'\in I$. 
\end{proof}

In Appendix \ref{app:staticAlgs} we discuss several well known static bin packing algorithms (and related definitions). We prove useful properties that are later used by our dynamic bin packing algorithms.

\section{Offline Scheduling Algorithms}\label{sec:offline}

In this section we design two offline algorithms proving Theorem \ref{thm:main-offline}. We start with the Covering Algorithm, which iteratively finds a set of requests that cover the time horizon and can be scheduled using one or two machines. We then present our Density-based Algorithm that finds dense subsets of requests of roughly the same duration.

\subsection{The Covering Algorithm}\label{sec:covering}

In this section we present Algorithm \ref{CoveringAlgorithm} whose performance is given by the following theorem.

\begin{thm} \label{thm:CoveringAlgorithm1}
The total cost of Algorithm \ref{CoveringAlgorithm} is at most $2 \cdot \|v\|_1$, for the uniform size case, $4 \cdot \|v\|_1$ in the non-uniform case.
If $\beta\leq \frac{1}{4}$ the total cost is at most $\sum_{t}\lceil\frac{2v_t}{1-2\beta}\rceil$.
\end{thm}

The main tool is the following idea of covers that are subsets of intervals that can be easily scheduled together. The proofs appear in Appendix \ref{app:covering}.

\begin{defn}\label{dfn-cover}
Given a set of intervals $\II$ with load vector $v$, a subset of intervals $C \subseteq\II$ is an $[\ell, u]$-cover if its load vector $v'$ satisfies that for any time $t$, $v'_t\in[\min\{v_t,\ell\},u]$.
\end{defn}

\begin{lemma}\label{lem:cover1}
Let $\II$ be a set of intervals. Then, it is possible to efficiently find
\begin{itemize}
    \item A $[1, 2]$-cover for the uniform size case.
    \item A $[\frac{1}{2}-\beta, 1]$-cover for the non-uniform size case when $\beta<\frac{1}{2}$.
\end{itemize}
\end{lemma}

Given Lemma \ref{lem:cover1} the algorithm is simple.

\begin{algorithm}
\SetAlgoLined
\DontPrintSemicolon
            {\bf In the non-uniform size case:} Schedule each interval with size greater than $\frac{1}{4}$ on a separate machine and remove it from $\II$.\;
            \While {$\II \neq \emptyset$} {
            Find a cover $C\subset \II$ as guaranteed by Lemma \ref{lem:cover1}. \;
            Schedule the intervals in $C$ using Algorithm \ref{FirstFitAlg} (First-Fit), and remove the intervals from $\II$.\;
            }
    \caption{Covering Algorithm} \label{CoveringAlgorithm}
\end{algorithm}

\subsection{Density-based Offline Algorithm}\label{sec:density-offline}

In this section we design our second algorithm whose cost is at most $c \cdot \|v\|_1 + O\left(\sum_{t=1}^{T}\sqrt{v_t\log \mu}\right)$ where $c=1$ in the uniform size case and $c=\min\left\{2,\frac{1}{1-\beta}\right\}$ in the non-uniform size case. We abuse here the notation of $v_t$ and define it as $\sum_{I \in \II(t)}w_I$ and not $\left\lceil \sum_{I \in \II(t)}w_I \right\rceil$. We prove the theorem with respect to these smaller values of $v_t$ (making the result only stronger). The algorithm is based on the following lemma that shows it is possible to find very dense packing whenever the load is large. The proofs appear in Appendix \ref{app:density-offline}.

\begin{lemma}\label{lem:density}
Let $\II$ be a set of intervals, and let $t$ be a time at which $v_t \geq 2 + 4 \ln \mu$. Then, it is possible to find efficiently a set $C\subseteq \II(t)$ such that  $\frac{1}{c} \leq \sum_{I\in C}w_I\ \leq 1$, and a length $\ell$ such that:
\begin{enumerate}
    \item The length of each interval $I \in C$ is at least $\ell$.
    \item All intervals in $C$ can be scheduled on a single machine of length at most $\ell\left(1+ 2\sqrt{\frac{2 + 4\ln\mu}{v_t}}\right)$.
\end{enumerate} 
\end{lemma}

Using Lemma \ref{lem:density} we design Algorithm \ref{DensityAlgorithm}.

\begin{algorithm}[h]
\SetAlgoLined
\DontPrintSemicolon
Let $\II$ be our current set of intervals. \;
\While { $\|v\|_{\infty}$ of the current set $\II$ is at least $2 + 4 \ln \mu$}{
     Apply Lemma \ref{lem:density} on $t_{\max} = \arg \max_t v_t$ to find subset of intervals $C \subseteq \II(t_{\max})$.  \;
     Schedule the intervals in $C$ on a single machine, and remove $C$ from $\II$.\;
  }
  Schedule the remaining intervals using Algorithm \ref{CoveringAlgorithm}.\;
\caption{Density Offline Algorithm} \label{DensityAlgorithm}
\end{algorithm}

\begin{thm}\label{thm:denseAlg}
The total cost of Algorithm \ref{DensityAlgorithm} is at most 
$$c \cdot \|v\|_1 + O\left(\sum_{t=1}^{T}\sqrt{v_t\log \mu}\right) \leq c \cdot \OPT + T\cdot O\left(\sqrt{\OA\log \mu}\right) $$
where $c=1$ in the uniform size case and $c=\min\left\{2,\frac{1}{1-\beta}\right\}$ in the non-uniform size case. 
\end{thm}

\begin{proof}
Consider an iteration $r$ of the loop of Algorithm \ref{DensityAlgorithm}. Let $\II^r$ be the current set of intervals with corresponding load values $v^r_t$, and let $C^r$ and $\ell^r$ be the subset of intervals and the length promised by Lemma \ref{lem:density}. By Lemma \ref{lem:density}, the total cost paid by the algorithm in this iteration is at most $\ell^r\left(1+ 2\sqrt{\frac{D}{\|v^r\|_{\infty}}}\right)$, where $D = 2 + 4 \ln \mu$. Let $\Delta v^r_t$ be the decrease in $v^r_t$ after removing the intervals in the subset $C^r$ from $\II^r$. Since the sum of sizes of intervals in $C^r$ is at least $1/c$, and the length of each interval is at least $\ell^r$, we get that $\sum_{t=1}^{T}\Delta v^r_t \geq \frac{1}{c} \cdot\ell^r$. Let $R$ be the total number of iterations in the loop. The total cost over all iterations is at most,
\allowdisplaybreaks
\begin{align}
\sum_{r=1}^R\ell^r\left(1+ 2\sqrt{\frac{D}{\|v^r\|_{\infty}}}\right) 
& \leq \sum_{r=1}^R c \sum_{t=1}^{T}\Delta v^r_t \cdot \left(1+ 2\sqrt{\frac{D}{\|v^r\|_{\infty}}}\right) \label{ine-dense1}\\
& \leq c \sum_{r=1}^R \sum_{t=1}^{T}\Delta v^r_t \cdot \min\left\{\left(1+ 2\sqrt{\frac{D}{v^r_t}}\right), 3 \right\}\label{ine-dense2}\\
& = c\sum_{r=1}^R \sum_{t=1}^{T}\Delta v^r_t \cdot\left( 1+ \min\left\{ 2\sqrt{\frac{D}{v^r_t}}, 2 \right\} \right).\nonumber
\end{align}
Inequality \eqref{ine-dense1} follows since $\sum_{t=1}^{T}\Delta v^r_t \geq \frac{1}{c} \cdot \ell^r$. Inequality \eqref{ine-dense2} follows since $\|v^r\|_{\infty} \geq v_t^r$, and since $\|v^r\|_{\infty} \geq D$ inside the loop.

Next, for each time $t$, we may analyze the summation $\sum_{r=1}^R \Delta v^r_t \cdot\left( 1+ \min\left\{ 2\sqrt{\frac{D}{v^r_t}}, 2 \right\} \right)$. Let $v_t= v_t^1$ be the starting value in the original instance $\II$. We get that,
\begin{align*}
    \sum_{r=1}^R \Delta v^r_t \cdot\left( 1+ \min\left\{ 2\sqrt{\frac{D}{v^r_t}}, 2 \right\} \right)
    & \leq v_t + \sum_{r=1}^R \Delta v^r_t \cdot\min\left\{ 2\sqrt{\frac{D}{v^r_t}}, 2 \right\} \\
     & \leq v_t + \sum_{r | v_t^r< D} 2 \Delta v^r_t + 
      2\sqrt{D} \sum_{r | v_t^r \geq D}\frac{\Delta v^r_t}{\sqrt{v_t^r}}\\
     & \leq v_t + 2D + 2\sqrt{D} \sum_{r | v_t^r \geq D}\frac{\Delta v^r_t}{\sqrt{v_t^r}}
\end{align*}

Finally, as $\frac{1}{\sqrt{v}}$ is a decreasing function of $v$ for $v > 0$, 
$\sum_{r | v_t^r \geq D}\frac{\Delta v^r_t}{\sqrt{v_t^r}} \leq \frac{1}{\sqrt{D}}+ 
    \int_{D}^{v_t} \frac{ dv}{\sqrt{v}}  
     = \frac{1}{\sqrt{D}} + 2\left(\sqrt{v_t} - \sqrt{D}\right)$. Plugging this, we get that the total cost of all iterations is at most $c \cdot \|v\|_1 + O\left(\sum_{t=1}^{T}\sqrt{v_t\log \mu}\right)$. 

Finally, by Theorem \ref{thm:CoveringAlgorithm1}, the total cost of Algorithm \ref{CoveringAlgorithm} is at most $4 \|v'\|_1$, where $v'$ is the final load vector (after applying all iterations). However, by the stopping rule of our algorithm we have $\|v'\|_{\infty} \leq  2 + 4 \log \mu$. Hence, the total additional cost is at most
\begin{equation*}
    4 \|v'\|_1 = 4 \sum_{t=1}^T \sqrt{v'_t} \cdot \sqrt{v'_t} \leq 4 \sum_{t=1}^T \sqrt{v'_t (2 + 4\log \mu)} \leq O(\sum_{t=1}^T \sqrt{v_t \log \mu})
\end{equation*}
Finally, using Jensen's inequality and substituting $v_{avg}\leq\OA$ we get that $\sum_{t=1}^T\sqrt{v_t\log\mu}\leq T\cdot\sqrt{\OA\log\mu}$, which concludes the proof.
\end{proof}

\subsection{Improving the Approximation For Non-Uniform Sizes}

In Appendix \ref{app:improve} we show how to draw on ideas from the analysis of the Harmonic Algorithm for static bin packing to improve the performance of algorithms for the dynamic bin packing problem. In particular, we partition the intervals into subsets based on their size, and schedule each subset separately using our algorithms. This proves the bounds in Theorem \ref{thm:main-offline}.

\section{Online Algorithm Using Lifetime and Average Load Predictions} \label{online_average}

In this section we design an algorithm having extra knowledge of the average load, which is a single value (the total load divided by the length of the time horizon). We start by presenting a transformation of certain static bin packing algorithms to the dynamic case, and then use it as a building block in the design of our Combined Algorithm. We complement this result with a lower bound when the lifetimes and average load are available to the algorithm. Finally, we discuss the effect of prediction errors, or noise, on the algorithms' guarantees.

\subsection{Transforming Static Bounded Space Algorithms to Non-Clairvoyant Dynamic Algorithms}\label{sec:online-transform}

In this section we show a general transformation of an online $k$-bounded space (static) bin packing algorithm to a non-clairvoyant online algorithm for the dynamic bin packing setting. We first define a static bin packing instance, given a dynamic bin packing instance, and prove an easy observation.
\begin{defn}\label{dy-to-stat}
Let $\II=\{I_1, \ldots, I_n\}$ be an instance of the dynamic bin packing problem where the size of interval $I_j$ is $w_{I_j}$. We define a corresponding static bin packing instance, $\II_S=\{i_1, \ldots, i_n\}$, such that for $j=1,\ldots, n: w_{I_j}=w_{i_j}$. Let $\OPT(\II)$ be the cost of the optimal solution for $\II$ and $\OPTS(\II_S)$ be the number of bins in an optimal solution for $\II_S$.
\end{defn}

\begin{obs} \label{observation_ineq}
For any instance $\II$, $OPT^{S}(\II_S)\leq \OPT(\II)$.
\end{obs}

\begin{proof}
Consider the instance $\II$. Obviously, shrinking all intervals to unit length can only decrease $\OPT$ without affecting $\OPTS(\II_S)$. Hence, we can assume that all intervals are of unit length. Next, consider machine $M$ active in the range $[s,t)$ in the optimal solution of $\II$. Let $I^j(M)$ be the set of intervals that have arrived in the range $[s+j-1,s+j)$ and are assigned to $M$. Notice that all items in $\II_S$ that correspond to intervals in $I^j(M)$ can be placed in a single bin in a solution of $\II_S$. Thus, all items in $\II_S$ corresponding to intervals that are assigned to machine $M$ can be assigned to $\lfloor t-s \rfloor$ bins. Hence, we can construct a feasible solution for instance $\II_S$ of cost (number of bins) no more than $\OPT(\II)$.
\end{proof}

Algorithm \ref{dynamic-alg} is given as input an online $k$-bounded static bin packing algorithm and applies it to the dynamic setting. The static bin packing instance is generated according to Definition \ref{dy-to-stat}.

\begin{algorithm}[h]
\SetAlgoLined
\DontPrintSemicolon
Let $A$ be $k$-bounded static bin packing algorithm (which maintains at any time at most $k$ active bins $b_1, b_2, \ldots, b_k$)\; 
Upon arrival of a new interval $I$: \;
\Begin{If $A$ opens a new bin, then open a new machine.\;
If $A$ accepts the item to bin $b_i$, accept the interval to machine  $m_i$.}

Upon departure of  an interval $I$: \;
\Begin{
If $I$ is not the last interval departing from its machine, do nothing.\;
If $I$ is the last interval departing from a non-active bin, close the machine.\;
If $I$ is the last interval departing from an active bin, close the machine and associate a new machine with the bin. Open the new machine if a new assignment to the respective active bin is made.\;
}
\caption{{\bf Dynamic non-clairvoyant algorithm\ }} \label{dynamic-alg}
\end{algorithm}

In the following lemma we analyze the cost of Algorithm \ref{dynamic-alg}. 

\begin{lemma}\label{lem-reduction}
Let $A$ be a $k$-bounded space bin packing algorithm that for instance $\II_S$ has cost at most $c \cdot OPT^{S}(\II_S) +\ell$.
Then, Algorithm  \ref{dynamic-alg} is an online non-clairvoyant algorithm for the dynamic bin packing setting whose cost on any instance $\II$ is at most:
\[c \cdot \mu\cdot  \OPT(\II) + \max\{k,\ell\} \cdot \|v\|_0.\] 
If $A$ is also $(c,\ell)$-decomposable (see Definition \ref{def:decomposable}), then its total cost when run separately on $n$ instances $\II_1, \ldots, \II_n$, where instance $\II_j$ has load vector $v^j$ and value $\mu_j$, is at most:
\[c \cdot \mu_{max}\cdot  \OPT(\II) +\max\{k,\ell\} \cdot \sum_{j=1}^{n} \|v^j\|_0,\] 
where $\II=\bigcup_{j=1}^{n}\II_j$ and $\mu_{max}=\max_{j=1}^{n}\mu_j$. 
\end{lemma}
The following is obtained by plugging Next-fit and the Harmonic algorithm into Lemma \ref{lem-reduction}.

\begin{corollary}\label{cor:dynamic-with-nf-harmonic}
Algorithm \ref{dynamic-alg} is a non-clairvoyant algorithm with total cost of at most:
\begin{itemize}
    \item $c\cdot\mu \cdot  \OPT(\II) + \|v\|_0$ when the underlying static bin packing algorithm is Next-fit, where $c=1$ if sizes are uniform; otherwise, $c=\min\left\{2,\frac{1}{1-\beta}\right\}$. 
    \item $\HH_k \cdot \mu \cdot \OPT(\II) + k \cdot \|v\|_0$ when the underlying static bin packing algorithm is Harmonic with parameter $k$.
\end{itemize}
\end{corollary}

\begin{proof}[Proof of Lemma \ref{lem-reduction}]

Consider an instance $\II$, and let $M$ be the set of machines that the algorithm opens due to $A$ opening a new bin. For the purposes of this analysis, when the last interval departs from an active bin, we consider the corresponding machine closed only if no other assignment is done on this bin. Otherwise, we consider the machine to be in a ``frozen'' state, where it is inactive (therefore not paying any cost), but it can become active again if a new assignment is made to the respective active bin. Each of these machines appears once in the set $M$. For each machine $m \in M$, let $s_m$ and $e_m$ be the times that the first and last interval is assigned to $m$ respectively, and let $f_m$ denote the duration for which $m$ remains frozen during the interval $[s_m, e_m)$ (can be zero if $m$ never becomes frozen).

Algorithm $A$ has at most $k$ active bins at any time $t$, and an interval is accepted to a machine only if $A$ accepts the item to the corresponding active bin. As a result, at any time there are at most $k$ accepting machines. We can therefore partition all machines $M$ into $k$ sets $P'_1, \dots, P'_k$, such that for $j=1,\dots,k$, any two machines $m_1, m_2 \in P'_j$ are not accepting at the same time.  It is obvious that such partition can also be found for any $k' \geq k$. For the following analysis, if $\ell \geq k $, we want to consider the partition into $\ell$ sets instead of $k$. So, for simplicity of exposition, we define $\alpha = \max\{k, \ell\}$, and consider the partition $P_1, \dots, P_{\alpha}$ such that any two machines $m_1, m_2 \in P_j$ are not accepting at the same time. Let $M_j$ be the number of machines in $P_j$.

Let $m_i^j$ be the machine with the $i$-th earliest start time among the machines in $P_j$. Let $s_i^j$ denote the start time of $m_i^j$, $e_i^j$ the time it accepts its last interval, and $f_i^j$ the duration for which it is frozen (can be zero). As a result, $m_i^j$ remains open at most during the interval $[s_i^j, \min\{e_i^j + \mu, T\})$ minus its freezing periods. Furthermore, it is obvious that $s_i^j \leq e_i^j$, $\forall i \in 1,\dots,M_j-1$, and by the properties of the algorithm, $e_i^j \leq s_{i+1}^j$, $\forall i \in 1,\dots, M_j-1$, since $m_{i+1}^j$ can become accepting only after $m_i^j$ has stopped accepting intervals, i.e. after time $e_i^j$. Finally, let $F_j = \sum_{i=1}^{M_j} f_i^j$  denote the total freezing time over all machines in $M_j$. The active machine-time required by the machines of each set $P_j$ is at most:
\begin{align*}
    \sum_{i=1}^{M_j} \left( \min\{T, e_i^j + \mu\}-s_i^j - f_i^j\right) & \leq T - s_{M_j}^j - F_j + \sum_{i=1}^{M_j-1} \left(e_i^j + \mu - s_i^j\right) \\
    & \leq \mu \cdot (M_j -1) + T-s_{M_j}^j -F_j + \sum_{i=1}^{M_j-1} \left(s_{i+1}^j - s_i^j\right) \\
    & \leq \mu \cdot (M_j -1) + T - s_1^j -F_j \leq \mu \cdot (M_j -1) + \|v\|_0.
\end{align*}

Summing up the costs for all $j = 1, \dots, \alpha$, the total cost of the algorithm is at most:
$$   \mu \cdot \sum_{j=1}^{\alpha} (M_j-1) + \alpha \cdot \|v\|_0. $$

Notice that the algorithm can open a new machine only when $A$ opens a new bin. Since by the guarantee of $A$ the optimal number of bins is at most $c \cdot  \OPTS(\II_S) + \ell$, we conclude that:
$$    \sum_{j=1}^{\alpha} M_j \leq c \cdot \OPTS(\II_S) + \ell. $$
Using this observation, we obtain that the total cost of the algorithm is at most:
\begin{align*}
    \mu \cdot \sum_{j=1}^{\alpha} (M_j-1) + \alpha \cdot \|v\|_0 & \leq \mu \cdot (c \cdot \OPTS(\II_S) + \ell - \alpha) + \alpha \cdot \|v\|_0 \\
    & \leq \mu \cdot c \cdot \OPT(\II) + \alpha \cdot \|v\|_0,
\end{align*}
where the last inequality follows from Observation \ref{observation_ineq}. This concludes the first part of the proof.

For the second part of the proof, let $A$ be a $(c, \ell)$-decomposable algorithm and assume that we run it separately on $n$ instances $\II_1, \dots, \II_n$, where $\II= \cup_{r=1}^n \II_r$. Let $v^r$ be the load vector of $\II_r$. We denote by $P_1^r, \dots, P_{\alpha}^r$ the partition of the machines opened in the solution of $\II_r$, and by $M_j^r$ the number of machines in $P_j^r$. Using previous arguments, the total cost of the solution for instance $\II_r$ is at most:
$$    \mu_r \cdot \sum_{j=1}^{\alpha} (M_j^r - 1) + \alpha \cdot \|v^r\|_0 $$

Since algorithm $A$ is $(c, \ell)$-decomposable, 
$$    \sum_{r=1}^n \sum_{j=1}^{\alpha} M_j^r \leq c \cdot \OPTS(\II_S) + n \cdot \ell $$

As a result, the total cost of the algorithm is at most:
\begin{align*}
    \sum_{r=1}^n \mu_r \cdot \sum_{j=1}^{\alpha} (M_j^r -1) + \alpha \cdot \sum_{r=1}^n \|v^r\|_0 & \leq \mu_{\max} \cdot (c \cdot \OPTS(\II_S) + n \cdot \ell - n \cdot \alpha) + \alpha \cdot \sum_{r =1}^n \|v^r\|_0 \\
    & \leq \mu_{\max} \cdot c \cdot \OPT(\II) + \alpha \cdot \sum_{r=1}^n \|v^r\|_0
\end{align*}
where again the last inequality follows from Observation \ref{observation_ineq}. This concludes the proof. 
\end{proof}

\subsection{Combined Algorithm}\label{sec:combined}

In this section we design an online algorithm (Algorithm \ref{algcombinedbinpacking}) that uses the lifetimes and average load information and whose total cost is at most $\HH_k\cdot OPT + T\cdot k \cdot O(\sqrt{v_{avg} \log \mu})$. The algorithm uses two parameters. Let $\epsilon_j = \min\{1, \frac{j}{v_{avg}}\}$. We say that an interval $I$ is in class $c^j$ if its predicted length $\ell_I\in [e^{\sum_{i<j}\epsilon_i}, e^{\epsilon_j} \cdot e^{\sum_{i<j}\epsilon_i}]$. Let $v^j$ be the load vector of intervals in class $c^j$. The second parameter used by the algorithm is $q_j=\min\{1, \sqrt{\frac{v_{avg}}{j}}\}$. Note that the values $\epsilon_j$ and $q_j$ depend only on $v_{avg}$ and not on $\mu$. 

\begin{algorithm}[h]
\SetAlgoLined
\DontPrintSemicolon
Hold a single copy of Algorithm \ref{FirstFitAlg} (First-Fit), and several copies of Algorithm \ref{dynamic-alg} with an underlying (static) $k$-bounded space $(c,\ell)$-decomposable algorithm (see Definition \ref{def:decomposable}).\;
Upon arrival of a new interval $I \in c^j$ at time $t$: \;
\uIf{$v^j_t \leq q_j$}  {Schedule the interval $I$ using the single copy of Algorithm \ref{FirstFitAlg} (First-Fit).}
      \Else{Schedule the interval $I$ using the $j$th copy of Algorithm \ref{dynamic-alg}.}
    \caption{Combined Algorithm (with $v_{avg}$ prediction)} \label{algcombinedbinpacking}
\end{algorithm}

Let $v^{j1}$ be the load vector of intervals $I\in c^j$ that were scheduled by Algorithm \ref{FirstFitAlg}. Let $v^{j2}$ be the load vector of intervals $I\in c^j$ that were scheduled be Algorithm \ref{dynamic-alg}. By this definition $v^j = v^{j1}+v^{j2}$. 
\begin{lemma}\label{lem:load-bounds}
For any $j$,
\begin{enumerate}
    \item $\|v^{j1}\|_{\infty} \leq q_j$.
    \item $\|v^{j2}\|_0 \leq (1+e^{\epsilon_j})|\{t \ | \ v^j_t \geq \frac{q_j}{2}\}|\leq  (1+e^{\epsilon_j})\sum_{t=1}^T\min\{1,\frac{2v^{j}_t}{q_j}\}.$
\end{enumerate}
\end{lemma}

\begin{proof}
We prove the two claims.

{\bf Proof of $(1)$:}
Consider any $t$, and take the intervals in $v_t^{j1}$. Let $I$ be the interval in $v_t^{j1}$ that arrived last, and let $t' \leq t$ be its start time. Since $I$ was scheduled using Algorithm \ref{FirstFitAlg}, $v^j_{t'} \leq q_j$.  At time $t'$, all other intervals in $v_t^{j1}$ are alive, and no other interval arrives in $v_t^{j1}$ between $t'$ and $t$, since $I$ was the last one. Therefore, $v_t^{j1} \leq q_j$, and since this holds for all $t$, $\|v^{j1}\|_{\infty} \leq q_j$.

{\bf Proof of $(2)$:} Since we are working with a single class $c^j$, we will remove the index $j$ and simply refer to it as $c$ and to its load vector as $v$. Let $[\ell, \ell \cdot e^{\epsilon})$ denote the lengths of intervals in class $c$. Let $c'$ be the set of intervals in $c$ that were scheduled by Algorithm \ref{dynamic-alg}, and let $v'$ be the load of $c'$. Finally, we use $R_1, R_2, \dots, R_r$ for the disjoint ranges (time intervals) in which the load is at least $\frac{q}{2}$ and there is an interval $I\in c'$ with $s_I \in R_i$. We use $|R_i|$ for the total duration of $R_i$.

By the behavior of the algorithm, for each interval $I \in c'$, we have $v_{s_I} > q$, i.e. $s_I \in R_i$ for some index $i$. Therefore, all $I \in c'$ can be associated with a range $R_i$ with $s_I \in R_i$. Let $c'_i$ be the set of intervals that are associated with range $R_i$. We prove that:
\[\|c'\|_0 \leq \sum_{i=1}^{r}\|c'_i\|_0 \leq \sum_{i=1}^r\left(|R_i|+\ell \cdot e^{\epsilon} \right)\leq (1+e^{\epsilon})\sum_{i=1}^r|R_i| \leq (1+e^{\epsilon}) |\{t \ | \ v^j_t \geq \frac{q_j}{2}\}|
\]

The first inequality follows since $\bigcup c'_i = c'$. The second inequality follows since each interval $I \in c'_i$ starts within $R_i$ and ends at most $\ell \cdot e^{\epsilon}$ after $R_i$'s end. The last inequality is true by definition of $R_i$, and we will now show that $|R_i| \geq \ell$ for all $i$ which proves the third inequality and completes this part of the proof.

Consider a range $R_i$ and take an interval $I \in c'_i$ with start time $s_I\in R_i$. Let $A \subseteq c$ be the intervals that are active at time $s_I$. Using Lemma \ref{intersecting-intervals} with $\alpha=q$, we see that there exists an interval $I \in A$ (with length at least $\ell$) that is active at time $t$ and sees more than $\frac{q}{2}$ load for its whole duration. This means that it is associated with range $R_i$ (otherwise the ranges would not be disjoint), and proves that $R_i$ has length at least $\ell$.

Finally, it is easy to see that $|\{t \ | \ v^j_t \geq \frac{q_j}{2}\}|=\sum_{t=1}^T \min\{1,\lfloor\frac{2v^{j}_t}{q_j}\rfloor\}$ which concludes the proof. 
\end{proof}

\begin{lemma}\label{lem:bound-n} Let $n$ be the maximal class such that $c^n \neq \emptyset$. Then,
$$n  = \left\{\begin{array}{ll} O(\sqrt{v_{avg} \log \mu}) & v_{avg} \geq 2\log \mu\\
O(\log \mu) & v_{avg}\leq 2\log \mu \end{array} \right.$$
\end{lemma}

\begin{proof}
The value $n$ satisfies that
$e^{\sum_{j=1}^{n-1}\epsilon_j} \leq \mu \leq e^{\sum_{j=1}^{n}\epsilon_j}$, which means $\sum_{j=1}^{n-1}\epsilon_j \leq \log \mu$.
By the choice of $\epsilon_j = \min\{1, \frac{j}{v_{avg}}\}$, we get: 
\begin{equation}
   \sum_{j=1}^{\min\{n, v_{avg}\}-1} \frac{j}{v_{avg}} + \sum_{v_{avg}}^{n-1} 1 = \min\{\frac{(n-1)n}{2 v_{avg}}, \frac{v_{avg}-1}{2}\}+ \max\{0, n-v_{avg}\} \leq \log \mu
  \label{logmu_ineq2}
\end{equation}

We now consider separately the two cases:

\begin{itemize}
\item 
{\bf When $v_{avg} \geq 2 \log \mu$:} In this case, $n \leq v_{avg}$, as if we assume the contrary, (\ref{logmu_ineq2}) leads to a contradiction. Therefore, (\ref{logmu_ineq2}) becomes $\log \mu \geq \frac{(n-1)n}{2v_{avg}}$. Rearranging this gives $n = O(\sqrt{v_{avg} \log \mu})$.
\item
{\bf When $v_{avg} \leq 2 \log \mu$:} If $n \leq v_{avg}$, then $n \leq 2 \log\mu$ and we are done. In the opposite case, (\ref{logmu_ineq2}) becomes:
\[
\frac{v_{avg}-1}{2} + n-v_{avg} \leq \log \mu \Rightarrow n \leq \log \mu + \frac{v_{avg}}{2} + \frac{1}{2} \leq 2 \log \mu + \frac{1}{2}
\]
and this concludes the proof.
\end{itemize}
\end{proof}

We are now ready to prove Theorem \ref{thm:combined-bin1}.

\begin{thm}\label{thm:combined-bin1}
Given an instance of Clairvoyant Bin Packing with average load prediction, the total cost of Algorithm \ref{algcombinedbinpacking} when it is executed with an underlying $k$-bounded space $(c,\ell)$-decomposable algorithm is at most 
$
 c\cdot \OPT + T \cdot \max\{k,\ell\}\cdot O(\sqrt{v_{avg} \log \mu})  
$. 
\end{thm}

\begin{proof}
The total cost of the algorithm is composed of two parts; the cost of Algorithm \ref{FirstFitAlg} (First-Fit) and the cost of all copies of Algorithm \ref{dynamic-alg}. Let $\II_j$ be the set of intervals scheduled by the $j$-th copy of Algorithm \ref{dynamic-alg} and $\II'=\cup_{j=1}^n\II_j$. The underlying static bin packing algorithm of Algorithm \ref{dynamic-alg} is $k$-bounded and $(c,\ell)$-decomposable. Thus, by Lemma \ref{lem-reduction} the total cost of scheduling $\II'$ is at most 
$$ c\cdot e^{\epsilon_n}\cdot \OPT(\II') +\max\{k,\ell\}\cdot\sum_{j=1 }^{n}\|v^{j2}\|_0 $$

The cost of Algorithm \ref{FirstFitAlg}\footnote{We remark that in the analysis of Algorithm \ref{FirstFitAlg} we took the worst case performance of $4$. This only affects the performance by additional constants.} is at most $4 T\cdot \|\sum_{j= 1}^{n}v^{j1}\|_{\infty}$. Combine the two bounds to get:
\begin{align}
\mbox{Cost of Alg} & \leq 4 T\cdot \|\sum_{j= 1}^{n}v^{j1}\|_{\infty} + c\cdot e^{\epsilon_n}\cdot \OPT(\II') +\max\{k,\ell\}\cdot\sum_{j=1 }^{n}\|v^{j2}\|_0 \nonumber\\
& \leq 4 T\cdot \sum_{j= 1}^{n} q_j + c\cdot e^{\epsilon_n} \OPT(\II)+\max\{k,\ell\}\cdot  \sum_{j=1}^{n}(1+e^{\epsilon_j})\sum_{t=1}^T\min\{1,\frac{2v^{j}_t}{q_j}\} \label{cost-alg2}\\
& \leq 4 T\cdot \sum_{j= 1}^{n} q_j + c\cdot e^{\epsilon_n} \OPT(\II)+  2(1+e)\max\{k,\ell\}\cdot\sum_{j=1}^{n}\min\{T,\frac{\|v^{j}\|_1}{q_j}\} \label{cost-alg3} \\
& \leq 4 T\cdot \sum_{j= 1}^{n} q_j + c\cdot e^{\epsilon_n} \OPT(\II)+  2(1+e)\max\{k,\ell\}\cdot T\cdot \min\{n,\frac{v_{avg}}{q_n}\} \label{cost-alg4}\\
& \leq c\cdot \OPT(\II) + T\cdot \max\{k,\ell\} \cdot O\left(\epsilon_n \cdot v_{avg} + \sum_{j= 1}^{n} q_j + \min\{n,\frac{v_{avg}}{q_n}\}\right).  \label{cost-alg5}
\end{align}

Inequality \eqref{cost-alg2} follows by using the fact that $\epsilon_j$ is increasing as $j$ is larger, and applying Lemma \ref{lem:load-bounds}. Inequality \eqref{cost-alg3} follows since $\epsilon_j\leq 1$ and using that $\sum_{t=1}^T\min\{1,\frac{v^{j}_t}{q_j}\} \leq \min\{T,\frac{\|v^{j}\|_1}{q_j}\}$. Inequality \eqref{cost-alg4} follows since $q_j$ is non-increasing in $j$. Finally, Inequality \eqref{cost-alg5} follows by rearranging and using that $\epsilon_j\leq 1$ and hence $e^{\epsilon_n}\leq 1+(e-1)\epsilon_n = 1+ O(\epsilon_n)$ and the fact that $\OPT\leq 4\|v\|_1$. We next analyze two cases:

{\bf Case 1, $v_{avg} \geq 2\log \mu$:}
In this case, by Lemma \ref{lem:bound-n}, $n= O(\sqrt{v_{avg} \log \mu})$. Substituting $n$ for this value, we get that $\epsilon_n =\min\{1, \frac{n}{v_{avg}}\} \leq \frac{n}{v_{avg}} = O(\sqrt{\frac{\log \mu}{v_{avg}}})$. Finally, as $q_j\leq 1$, $\sum_{j=1}^{n}q_j \leq n =O(\sqrt{v_{avg} \log \mu})$.
Plugging these bounds into Inequality \eqref{cost-alg5} we get the desired result.

{\bf Case 2, $v_{avg} \leq 2\log \mu$:} In this case $v_{avg}= O(\sqrt{v_{avg} \log \mu})$.
By the Lemma \ref{lem:bound-n}, $n=O(\log \mu)$.  
Using this bound we get that:
$q_n =\min\{1,\sqrt{\frac{v_{avg}}{n}}\}$. Thus,
$\frac{v_{avg}}{q_n}\leq \max\{v_{avg}, O(\sqrt{v_{avg} \log \mu})\}= O(\sqrt{v_{avg} \log \mu})$.
Finally,
\[\sum_{j=1}^{n}q_j \leq \sum_{j=1}^{v_{avg}}q_j + \sum_{j=v_{avg}}^{n}q_j\leq
v_{avg}+ \sum_{j=v_{avg}}^{n}\sqrt{\frac{v_{avg}}{j}}\leq v_{avg} + \sqrt{v_{avg}\cdot n} = O(\sqrt{v_{avg} \log \mu})\]
Using that $\epsilon_n\leq 1$, and plugging everything into Inequality \eqref{cost-alg5} we get the desired result.
\end{proof}

As seen in Lemma \ref{lem:static-decomposable} the Next-Fit and Harmonic algorithms are decomposable (see Definition \ref{def:decomposable}). Thus, using the Next-Fit algorithm (in the uniform size case) and the Harmonic algorithm (in the non-uniform size case) as the underlying algorithms of Algorithm \ref{dynamic-alg} produces the following corollary.

\begin{corollary}
The total cost of Algorithm \ref{algcombinedbinpacking} is at most:
\begin{itemize}
    \item $OPT+T\cdot O(\sqrt{v_{avg}\log\mu})$ (uniform size).
    \item $\HH_k\cdot OPT+T\cdot k\cdot O(\sqrt{v_{avg}\log\mu})$ (non-uniform size).
\end{itemize}
\end{corollary}

\subsection{Lower Bound: Dynamic Clairvoyant Bin Packing}\label{app:lower-bound}

In this section we complement the results of Section \ref{sec:combined} by generalizing the lower bound of \cite{azar2017tight}  to take into account also $\OA$ showing that the additive term $O(\sqrt{\OA \cdot \log \mu})$ is indeed unavoidable, if only the average future load and lifetime is available to an online algorithm.  

\begin{lemma}
For any values $\mu, v_{avg}$ the total cost of any algorithm is at least:
$$\Omega\left(T\cdot\sqrt{v_{avg}\log \mu}\right) = \Omega\left(T\cdot\sqrt{\OA \cdot \log \mu}\right).$$
\end{lemma}

\begin{proof}

First if $\frac{\log \mu}{v_{avg}}\leq 2$. Then the bound is meaningless since in this case the cost of $OPT$ is at least $\|v\|_1 \cdot \Omega(1) = \Omega\left(T\cdot\sqrt{v_{avg}\log \mu}\right)$.
Otherwise; $v_{avg}< \frac{\log \mu}{2}$, and we show an adversary that given a parameter $a \gets v_{avg}$ (the desired average load) and $\mu$, creates an instance such that the average load is always at least $a$ and the algorithm pays at least $\|v\|_1 \cdot \sqrt{\frac{\log \mu}{2a}}$. If the actual average load of the instance is strictly more than $a$ (the desired average load), we can extend the time horizon without adding new requests until the average load drops to the desired average, $a$. Of course, this extension does not affect the total cost. The adversary initiates the following sequence:

\begin{itemize}
    \item at each time $t=1, \ldots, \mu$ as long as the algorithm has strictly less than $N =\sqrt{2a \log \mu}$ active machines:
    \item Initiate sequentially requests of size $w=\sqrt{\frac{2a}{\log \mu}}\leq 1$ (since $a \leq \frac{\log \mu}{2}$) and of increasing length of $2^i$, $i=0,1,\ldots, \lceil\log \mu\rceil$.
\end{itemize}

Since each request is of length at most $\mu$, the total length of the time horizon $T\leq 2\mu$ (and this adversarial sequence can be repeated again afterwards).

First, the adversary indeed manages to make the algorithm open at least $N=\sqrt{2a \log \mu}$ machines at each time $t$ since otherwise the load of the requests initiated at time $t$ is at least $\lceil\log \mu\rceil\cdot w \geq  \log \mu \cdot \sqrt{\frac{2a}{\log \mu}} = N$.  Hence, the average load at each time $t=1, \ldots, \mu$ is at least $w\cdot N=2a$, and, since the length of the time horizon of the sequence is at most $2\mu$, the average load over the whole horizon is at least $a$ as promised. 

Let $\ell_{t}$ be the longest interval that the adversary releases at time $t$ ($0$ if there is no such interval). We have that: 
\begin{align}
     \|v\|_1 & \leq 2\sum_{t=1}^{\mu} w \cdot \ell_t  \label{ineq111}\\
     & \leq 2w \cdot c_{alg} = c_{alg} \cdot  2\sqrt{\frac{2a}{\log \mu}} \label{ineq113}
\end{align}
Inequality \eqref{ineq111} follows since the total size of the last interval in the round dominates all previous ones at that round by the geometric power of $2$ (the longest item dominates the rest).  Inequality \eqref{ineq113} follows by the observation that the algorithm opens a new machine for the last interval the adversary gives at a certain round. Hence, $c_{alg} \geq \sum_{t=1}^{\mu} \ell_t$.
Rearranging, we get that $c_{alg} \geq \|v\|_1 \cdot \Omega\left(\sqrt{\frac{\log \mu}{v_{avg}}}\right) = \Omega\left(T\cdot\sqrt{v_{avg}\log\mu}\right)$. Lastly, Theorem \ref{thm:CoveringAlgorithm1} states that $\OA\leq 4\cdot v_{avg}$ which concludes the proof.

\end{proof}

\subsection{Handling Inaccurate Predictions}\label{sec:inaccurate}
So far, we have assumed that predictions for either interval lengths or average load are accurate (or ``noiseless"). Naturally, some predictions are in practice prune to errors. In this section we examine the performance of Algorithm \ref{algcombinedbinpacking} in the presence of prediction errors. We show that the algorithm is robust to prediction errors with respect to both ($v_{avg}$ and interval lengths). Formally,  

\begin{thm} \label{thm:noise}
Suppose that Algorithm \ref{algcombinedbinpacking} is given predicted values $v'_{avg}$ and interval lengths $\ell'_I$ such that $v'_{avg}\in \left[\frac{v_{avg}}{1+\delta}, v_{avg} \cdot(1+\delta)\right]$, and each $\ell'_I \in \left[\frac{\ell_I}{1+\alpha},\ell_I \cdot(1+\lambda)\right]$. Then, its total cost is at most 
\begin{itemize}
    \item $(1+\alpha)\cdot(1+\lambda) \cdot (OPT+T\cdot O(\sqrt{(1+\delta) v_{avg}\log\mu}))$ (uniform size).
    \item $(1+\alpha)\cdot(1+\lambda)\cdot ( \HH_k \cdot OPT +Tk\cdot O(\sqrt{(1+\delta) v_{avg}\log\mu}))$ (non-uniform size).
\end{itemize}
\end{thm}

\begin{proof}
We first show how to handle the inaccuracy in predicting $v_{avg}$. We claim that for the purpose of analysis (with loss in performance) we can assume that our prediction $v'_{avg}$ is accurate. Indeed, if  $v'_{avg}<v_{avg}$, we can extend the time horizon with no additional requests from $T$ to $T'$ such that $T'\cdot v'_{avg} = T\cdot v_{avg}$. This makes the average load $v'_{avg}$ (and clearly does not change the costs of the algorithm and OPT). However, the additive cost increases to $T'\cdot O(\sqrt{v'_{avg}\log\mu})=T\cdot O(\sqrt{((1+\delta)v_{avg}\log\mu})$.
If $v'_{avg}> v_{avg}$, we can add fictitious $T(v_{avg}'-v_{avg})$ intervals of length 1 at the end of the time horizon. This increases (for the analysis) the cost of both the algorithm and OPT by this value. Again, this increases the actual load to $v'_{avg}$, and the additive term becomes $T\cdot O(\sqrt{v'_{avg}\log\mu})= T\cdot O(\sqrt{(1+\delta)v_{avg}\log\mu})$.

To analyse the errors in the length predictions we observe that in this case the analysis of Algorithm \ref{algcombinedbinpacking} is almost unchanged. There are only two modifications to be made. First, we generalize Lemma \ref{lem:load-bounds} as follows. For any $j$,
$$\|v^{j2}\|_0 \leq (1+\alpha)(1+\lambda)(1+e^{\epsilon_j})|\{t \ | \ v^j_t \geq \frac{q_j}{2}\}|\leq  (1+\alpha)(1+\lambda)(1+e^{\epsilon_j})\sum_{t=1}^T\min\{1,\frac{2v^{j}_t}{q_j}\}.$$
Second, the maximum length ratio of all intervals scheduled by each copy of Algorithm \ref{dynamic-alg} grows by the length prediction error to at most $e^{\epsilon_n}\cdot(1+\alpha)(1+\lambda)$. Thus, the total scheduling cost is at most:

$$ c(1+\alpha)(1+\lambda)\cdot e^{\epsilon_n}\cdot \OPT(\II') +\max\{k,\ell\}\cdot\sum_{j=1 }^{n}\|v^{j2}\|_0. $$
\end{proof}
Note that as a special case of the above theorem, the algorithm pays a rather small multiplicative factor $(1 + \delta)$ when the only inaccurate prediction is in the average load. This is significant in other domains of interest in which the interval lengths are precisely known upon arrival (e.g., establishing virtual network connections, see \cite{azar2017tight}). Unfortunately, in the general case, Theorem \ref{thm:noise} implies that the algorithm pays (at the worst case) an additional cost which is proportional to the noisiest length prediction. This result is tight as suggested by the lemma below.

\begin{lemma}
Given length prediction $\ell'_I$ such that $\ell'_I\in\left[\frac{\ell_I}{1+\alpha},\ell_I\cdot(1+\lambda)\right]$, the cost of any online algorithm is at least $(1+\alpha)(1+\lambda)OPT$.
\end{lemma}

\begin{proof}
At time $0$ the adversary adds $k^2$ intervals with predicted length $(1+\lambda)$ and width $\frac{1}{k}$. All these intervals are located on at least $k$ machines. The true length of each interval is as follows: on each machine there is exactly one interval of length $(1+\alpha)(1+\lambda)$ and the length of the rest of the intervals is $1$.

The total cost of the algorithm is at least $(1+\alpha)(1+\lambda)k$ as there are at least $k$ active machines. The cost of the optimal solution is $(1+\alpha)(1+\lambda) + k - 1$. For $k$ large enough compared to the prediction error of the length, we get that the ratio between the cost of any algorithm and the optimal solution is,
$$ \frac{(1+\alpha)(1+\lambda)k}{(1+\alpha)(1+\lambda) + k - 1}\rightarrow (1+\alpha)(1+\lambda). $$
\end{proof}

\section{Online Algorithm Using Lifetime and Load Vector Predictions} \label{full_predictions}

In this section we design an online algorithm with an extra knowledge of the future load. In particular, at each time $t$ the algorithm is given the value $v_{t'}$ for all $t'\in [t,t+\mu)$.
Without loss of generality (by refining the discretization of the time steps), we assume that at most one interval arrives at any time $t$. 

Algorithm \ref{alg:single-cover-online} shows how to construct a single cover, similar in nature to the covers offline Algorithm \ref{CoveringAlgorithm} is creating, but in an online fashion. To this end, it takes as an input a load vector $v'$ which might be an \textbf{overestimate load prediction}\footnote{This can be an overestimate due to prediction errors, however, the algorithm itself later produces overestimates by design (and not due to errors), which are used recursively in our analysis.} of the real load $v$, i.e. $v'\geq v$, and let $\Delta = v'-v \geq 0$. When an interval arrives, Algorithm \ref{alg:single-cover-online} decides whether to accept it or reject it, with accepted intervals becoming part of the cover. Algorithm \ref{algorithm-online-cover} then uses Algorithm \ref{alg:single-cover-online} to create a set of covers and schedule all intervals to machines.

\begin{algorithm}[h]
\SetAlgoLined
\DontPrintSemicolon
Let $v'\geq v$ be an overestimate load prediction given to the algorithm.\;
When interval $I$ arrives at time $t'$.
Let $v^a(t'), v^r(t')$ be the load vector of the intervals that arrived prior to time $t'$, and were accepted or rejected respectively. \;
Accept $I$ if there is a time $t\in I$ such that,
\[v'_t-v^r_t(t')\leq \left\{\begin{array}{ll} 1 & \mbox{in the uniform size case}\\
\frac{1}{2} & \mbox{in the non-uniform size case}\end{array}\right.\] 
Otherwise, reject $I$.
\caption{Online covering with overestimate predictions\ }
\label{alg:single-cover-online}
\end{algorithm}

\begin{lemma}\label{lem-online-load1-unified}
Let $v^a$ be the intervals that Algorithm \ref{alg:single-cover-online} accepted. Then, for each time $t$:
\begin{align*}
    \min\{(1-\Delta_t)^+, v_t\} & \leq v^a_t \leq 2 & \mbox{ for the uniform size case}\\
    \min\{(\frac{1}{2}-\beta-\Delta_t)^+, v_t\} & \leq v^a_t \leq 1 & \mbox{ for the non-uniform size case}
\end{align*}
\end{lemma}

\begin{proof}  We first show that the upper bounds hold, and then provide a proof for the lower bounds.

\paragraph*{Upper bounds:}
Suppose there is a time $t$ where $v_t^a$ exceeds the upper bound. We will show that the algorithm cannot have accepted all intervals in $v^a$. Consider the intervals $I$ that belong to $v^a$ and which are active at time $t$, i.e., $t \in I$. Let that set of intervals be denoted by $A$, and its load with $v''$. Since all intervals in $A$ are active at time $t$, we can apply Lemma \ref{intersecting-intervals}. 

\emph{For the uniform case:} Using Lemma \ref{intersecting-intervals} with $v=2$, we can see that there is at least one interval that observes load strictly more than $1$ for its whole duration. 
Let $I$ be the first such interval, and let $s_I$ be its arrival time.
The intervals in $v^a$ are never rejected by the algorithm and hence, $v_{t}^r(s_I) \leq v_t-v_t^a$. Hence, upon arrival of $I$, for any time $t \in I$ $v'_t - v^r_{t}(s_I) \geq v_t - v^r_{t}(s_I) \geq v_t^a > 1$, and $I$ will not be accepted by the algorithm.

\emph{For the non-uniform case:} Similarly, using Lemma \ref{intersecting-intervals} with $v=1$, we see that there is at least one interval that observes load greater than $\frac{1}{2}$ for its whole duration. Let $I$ be the first such interval, and let $s_I$ be its arrival time. The intervals in $v^a$ are never rejected by the algorithm and hence, $v_{t}^r(s_I) \leq v_t-v_t^a$. Hence, upon arrival of $I$, for any time $t \in I$ $v'_t - v^r_{t}(s_I) \geq v_t - v^r_{t}(s_I) \geq v_t^a > \frac{1}{2}$, and $I$ will not be accepted by the algorithm. 

\paragraph*{Lower bounds:}
This proof is also by contradiction. Let $t$ be a time for which $v^a_t$ is smaller than the lower bound. Let $I$ be the last arrived interval that is active at time $t$ and which was rejected by the algorithm. Upon $I$'s arrival at time $s_I$, the current load vector of rejected intervals at time $t$ is $v_t^r(s_I)$ (does not yet include $I$), and so, $v_t^r(s_I) + v_t^a + w_I = v_t$. Upon $I$'s arrival, the algorithm considers the quantity $v_t' - v_t^r(s_I)$.

\emph{For the uniform case:} The assumption that the lower bound is violated is translated to $v^a_t < \min\{(1-\Delta_t)^+, v_t\}$.  If $\Delta_t \geq 1$, this leads to an obvious contradiction, since it would imply that $v_t^a < 0$. Therefore, we focus on the case $\Delta_t < 1$, so $(1-\Delta_t)^+ = 1 - \Delta_t$. This gives: $v_t' - v_t^r(s_I) = v_t'-v_t + v_t^a +w_I = \Delta_t + v_t^a + w_I < \Delta_t + \min\{1-\Delta_t, v_t\} + w_I = \min\{1, v_t + \Delta_t\} + w_I \leq 1 + w_I$. Since all intervals have size $1/\uni$ for some integer $\uni$, both the overestimated load $v_t'$ and the rejected load $v_t^r(s_I)$ are multiples of $1/\uni$ (if $v_t'$ is not, it can be rounded down to the closest $1/\uni$ multiple, since we know the extra load does not correspond to some interval). Therefore, since $v_t' - v_t^r(s_I) < 1 + w_I$ and both loads are multiples of $1/\uni$, we have $v_t' - v_t^r(s_I) \leq \frac{\uni-1}{\uni} + \frac{1}{\uni} = 1$. This shows that $I$ is accepted by the algorithm and leads to a contradiction proving that $v^a_t \geq  \min\{(1-\Delta_t)^+, v_t\}$ for each time $t$.

\emph{For the non-uniform case:}
We assumed that $v^a_t < \min\{(\frac{1}{2}-\beta-\Delta_t)^+, v_t\}$. If $\Delta_t \geq \frac{1}{2} - \beta$, this leads to an obvious contradiction, since it would imply that $v_t^a < 0$. Therefore, we focus on the case $\Delta_t < \frac{1}{2} - \beta$, so $(\frac{1}{2}-\beta-\Delta_t)^+ = \frac{1}{2}-\beta-\Delta_t$. We then have $v_t' - v_t^r(s_I) = v_t'-v_t + v_t^a +w_I = \Delta_t + v_t^a + w_I < \Delta_t + \min\{\frac{1}{2}-\beta-\Delta_t, v_t\} + w_I = \min\{\frac{1}{2}-\beta, v_t + \Delta_t\} + w_I \leq \frac{1}{2}-\beta + \beta = \frac{1}{2}$. Therefore, $I$ is accepted by the algorithm leading to a contradiction that proves that $v^a_t \geq  \min\{(\frac{1}{2}-\beta-\Delta_t)^+, v_t\}$ for each time $t$.
\end{proof}

We now present the Online Covering Algorithm. This algorithm uses copies of Algorithm \ref{alg:single-cover-online} to create a set of covers online and schedule all intervals to machines. 

\begin{algorithm} 
\SetAlgoLined
\DontPrintSemicolon
            {\bf In the non-uniform size case:} Schedule each interval with size greater than $\frac{1}{4}$ on a separate machine.\;
Run copies $i=1,2, \ldots$ of the online covering algorithm with overestimate predictions (Algorithm \ref{alg:single-cover-online}). The $i$th copy receives an overestimate of
\[v'^i_t= \left\{\begin{array}{ll} (v_t-(i-1))^+ & \mbox{in the uniform size case}\\
(v_t-(i-1) \cdot (\frac{1}{2}-\beta))^+ & \mbox{in the non-uniform size case}\end{array}\right.\] 
The $i$th copy of the algorithm receives as its input all intervals that are rejected from copies $1,2, \ldots, i-1$.\;
Schedule all intervals accepted by copy $i$ using Algorithm \ref{FirstFitAlg} (First Fit).\;
  \caption{Online Covering Algorithm (with load vector predictions)}
   \label{algorithm-online-cover}
\end{algorithm}

\begin{lemma} \label{induction-lemma-unified}
Let $v^a_{t,i}$ be the total load of accepted intervals of copy $i$ at time $t$. Let $v^w$ and $v^n$ be the load vector of the intervals that have size more than $\frac{1}{4}$ and at most $\frac{1}{4}$ respectively, and let $\beta_n$ be the largest size of intervals in $v^n$. For every time $t$, and $j$:

\begin{align*}
    \sum_{i=1}^{j}v^a_{t,i} & \geq \min\{j ,v_t\} & & \mbox{ for the uniform size case}\\
    \sum_{i=1}^{j}v^a_{t,i} & \geq \min\{(j \cdot (\frac{1}{2}-\beta_n)-v_t^w)^+,v^n_t\} & & \mbox{ for the non-uniform size case}
\end{align*}
\end{lemma}

\begin{proof}
     The proof is by induction on $j$.
For the first copy, $j=1$, we have $\Delta_t=0$ for the uniform size case and $\Delta_t = v_t^w$ for the non-uniform size case, where $v_t^w$ is the total width of intervals of sizes larger than $1/4$. For $j=1$ the claim follows by Lemma \ref{lem-online-load1-unified}. For time $t$, let $v_1, v_2, \ldots, v_{j-1}$ be the load accepted by copies $1,2, \ldots, j-1$. 

\emph{For the uniform case:}
By our guarantee, $\sum_{i=1}^{j-1}v_i \geq \min\{j-1, v_t\}$. We assume that 
$v_t > \sum_{i=1}^{j-1}v_i$, and $\sum_{i=1}^{j-1}v_i < j$, otherwise we are done. For the last copy the actual load at time $t$ is $v''= v_t-\sum_{i=1}^{j-1}v_i$. Hence, $\Delta_t = \sum_{i=1}^{j-1}v_i- (j-1)$ (which is greater than 0 by the induction hypothesis), and it is guaranteed to accept at least $\min\{(1-\Delta_t)^+, v''_t\}$ load by Lemma \ref{lem-online-load1-unified}. Therefore, the total load of accepted intervals of the first $j$ copies is at least:

\[\sum_{i=1}^{j-1}v_i + \min\{(1-\Delta_t)^+, v''_t\} = 
\sum_{i=1}^{j-1}v_i + \min\{(j-\sum_{i=1}^{j-1}v_i)^+, v_t-\sum_{i=1}^{j-1}v_i\} = \min\{j,v_t\}\]

\emph{For the non-uniform case:} Similarly to the uniform case, by our guarantee, $\sum_{i=1}^{j-1}v_i \geq \min\{((j-1) \cdot (\frac{1}{2} - \beta_n)-v_t^w)^+, v^n_t\}$. We assume that 
$v^n_t > \sum_{i=1}^{j-1}v_i$, and $\sum_{i=1}^{j-1}v_i < j\cdot (\frac{1}{2}-\beta_n)-v_t^w$, otherwise we are done. For the last copy the actual load at time $t$ is $v''= v^n_t-\sum_{i=1}^{j-1}v_i$. Hence, $\Delta_t = \sum_{i=1}^{j-1}v_i- (j-1)\cdot(\frac{1}{2}-\beta_n)+v_t^w$, and it is guaranteed to accept at least $\min\{(\frac{1}{2}-\beta_n-\Delta_t)^+, v''_t\}$ load by Lemma \ref{lem-online-load1-unified}. Therefore, the total load of accepted intervals of the first $j$ copies is at least:
\begin{align*}
 \sum_{i=1}^{j-1}v_i + \min\{(\frac{1}{2}- \beta_n-\Delta_t)^+, v''_t\} & = 
\sum_{i=1}^{j-1}v_i + \min\{(j\cdot(\frac{1}{2}-\beta_n)-v_t^w-\sum_{i=1}^{j-1}v_i)^+, v^n_t-\sum_{i=1}^{j-1}v_i\} \\
& \geq \min\{(j\cdot(\frac{1}{2}-\beta_n)-v_t^w)^+,v^n_t\} 
\end{align*}
\end{proof}

\begin{thm} \label{thm:online-full-pred}
Given an instance of Clairvoyant Bin Packing with load vector predictions, the total cost of Algorithm \ref{algorithm-online-cover} is at most $2 \cdot \|v\|_1$, for the uniform size case, $8 \cdot \|v\|_1$ in the non-uniform case.
If $\beta\leq \frac{1}{4}$ the total cost is at most $\sum_{t}\lceil\frac{2v_t}{1-2\beta}\rceil$.
\end{thm}

\begin{proof}
     \emph{For the uniform case:} From Lemma \ref{induction-lemma-unified}, for each time $t$, the first $v_t$ copies will accept all intervals that are active at time $t$, and according to Lemma \ref{thm-ff}, each copy can schedule its intervals using 2 machines. As a result, the total cost of the algorithm is $\sum_t 2v_t = 2\|v\|_1$, proving 2-competitiveness. 
     
     \emph{For the non-uniform case:} If $\beta \leq \frac{1}{4}$, then $v_t^n=v_t$ and $v_t^w=0$. Then, by Lemma \ref{induction-lemma-unified}, the first $\lceil\frac{v_t}{\frac{1}{2}-\beta}\rceil$ copies will accept all intervals that are active at each time $t$. By Lemma \ref{lem-online-load1-unified}, each of them is paying 1, so the total cost is at most $\sum_t\lceil\frac{2v_t}{1 - 2\beta}\rceil$. 
     
     If $\beta>\frac{1}{4}$, let $W_t$ be the number of intervals with size larger than $\frac{1}{4}$ that are active at time $t$. Since the algorithm opens a separate machine of unit size for each of them, it pays cost $W_t$ at each time $t$. The cost at time $t$ to schedule intervals of size smaller than $\frac{1}{4}$ is at most $\lceil \frac{v^w_t+v^n_t}{\frac{1}{2}-\frac{1}{4}}\rceil= \lceil 4 v_t\rceil$. Thus, the total cost is at most: $~W_t+ \lceil4 v_t\rceil = \lceil W_t + 4v_t \rceil  \leq \lceil 4v^w_t + 4 v_t\rceil \leq 8\lceil v_t\rceil.$
\end{proof}

\section{Conclusion}
This paper studies the VM allocation problem in both offline and online settings. Our main contribution is the design of novel algorithms that use certain predictions about the load (either its average or the future time series). We show that this extra information leads to substantial improvements in the competitive ratios. As future work, we plan to consider additional models of prediction errors and examine the performance on real-data simulations.

\bibliographystyle{acm}

\bibliography{ms}

\appendix

\section{Static Bin Packing Algorithms}\label{app:staticAlgs}

We use several well known online static bin packing algorithms, or variants of them for the dynamic setting. In the (static) bin packing case, that has been studied extensively, items do not depart, and the goal is to minimize the total number of bins used.

First, the well known First-Fit algorithm appears as Algorithm \ref{FirstFitAlg}.

\begin{algorithm} 
\SetAlgoLined
\DontPrintSemicolon
When an interval $I$ arrives at time $t$,
assign it to a machine with the earliest opening time among the available machines. If no machine is available, open a new machine.
\caption{First-Fit \ } \label{FirstFitAlg}
\end{algorithm}

\begin{lemma}
\label{thm-ff}
The total cost of Algorithm \ref{FirstFitAlg} is:
\begin{enumerate}
    \item $\|v\|_{\infty} \cdot \|v\|_0$ for the uniform size case. \label{ff-1}
\item $\left(\frac{1}{1-\beta} \cdot \|v\|_{\infty}+1\right) \cdot \|v\|_0$ for the non-uniform size case when $\beta\leq\frac{1}{2}$. \label{ff-2}
    \item $4 \cdot \|v\|_{\infty} \cdot \|v\|_0$ for the non-uniform size case when $\beta>\frac{1}{2}$. \label{ff-3}
    \item $\|v\|_0$ if $\|v\|_{\infty}\leq 1$. \label{ff-4}
\end{enumerate}
\end{lemma}

\begin{proof}
The proof for each case follows.

\noindent {\bf Uniform case.}
Assume in contradiction that there exists a time $t$ in which an interval $I$ arrives, resulting in $\|v\|_{\infty}+1$ open machines. This can only happen if the  $\|v\|_{\infty}$ machines that were open prior to the arrival of $I$ (at time $t$) are all fully occupied. Along with the interval $I$, by the properties of first fit, it  implies that the number of intervals at time $t$ is greater than $N_t$, a contradiction. 

\noindent {\bf Non uniform case when $\beta\leq\frac{1}{2}$.}
Let $M$ denote the maximum number of machines used at any time over the horizon, and let $t$ be a time where the algorithm decided to open an $M$th machine. Since the size of each interval is at most $\beta$ and at time $t$ the algorithm could not accommodate an arriving interval in the existing $M-1$ machines, that implies that the $M-1$ first machines have load at least $1-\beta$ at time $t$. Furthermore, the total load of machines $M-1$ and $M$ at time $t$ is at least 1, otherwise the algorithm would not need to open a new machine. As a result, at time $t$: 
\begin{equation*} 
    \|v\|_{\infty} > (M-2) \cdot (1-\beta) + 1~\Longrightarrow ~ \frac{1}{1-\beta} \cdot \|v\|_{\infty} \geq M -2 + \frac{1}{1-\beta} \geq M -1. 
\end{equation*}
Therefore,
$  M \leq \frac{1}{1-\beta} \cdot \|v\|_{\infty} + 1$, as claimed.

\noindent {\bf Non uniform case when $\beta > \frac{1}{2}$.} 
Fix a time $t$. A machine is called {\em wide} if it has load at least $1/2$ and {\em narrow} otherwise. Denote by $M_w$ and $M_n$  the number of wide and narrow machines, respectively. We show that $M_w \leq 2\cdot \|v\|_\infty$ and $M_n \leq 2\cdot \|v\|_\infty$.
\begin{itemize}
    \item Wide machines: Assume $M_w \geq 2 \cdot \|v\|_\infty +1$. The wide machines have load at least $1/2$, thus their total load is at least $ \frac{1}{2} M_w \geq \|v\|_\infty + \frac{1}{2} > \|v\|_\infty$, which is a contradiction. Thus, $M_w \leq 2\cdot \|v\|_\infty$.
    \item Narrow machines: Assume $M_n \geq 2 \cdot \|v\|_\infty  + 1$ and let $m$ denote the last machine activated at $t$, among the narrow machines. At time $t$, by definition, $m$ has at least one interval $I$ with size less than $1/2$. At the  start time $s_I$ of $I$, the algorithm chose to assign it to machine $m$, implying that $I$ could not be assigned to all other narrow machines, meaning they each had load at least $1/2$ at time $s_I$. Hence, the total load on the narrow machines at time $s_I$ is at least $(M_n-1)\frac{1}{2} + w_I \geq \|v\|_\infty + w_I > \|v\|_\infty$, which is a contradiction. 
\end{itemize}
As a result, at time $t$ the total number of open machines is at most $M_w + M_n \leq 4\|v\|_\infty$, and, since this is true for all $t$, we get the desired result. 

\noindent
{\bf When $\|v\|_{\infty}\leq 1$:}
If the algorithm opens more than a single machine at any time $t$, the total load at that time is more than $1$, contradicting our assumption.
\end{proof}

An online static bin packing algorithm is said to be {\em $k$-bounded-space} if, for each new item, the number of bins in which it can be packed is at most $k$. 

The Next-Fit algorithm is a prime example of a bounded space algorithm. It holds exactly one active bin at any time. Upon arrival of an item that does not fit in the active bin, it closes it and opens a new one (in which the new item is placed). Thus, the Next-Fit algorithm is $1$-bounded. 

Another important example of a bounded space algorithm is the Harmonic algorithm \cite{lee1985simple}. The $k$-bounded space Harmonic algorithm partitions the instance $\II=\cup_{j=1}^{k}\II_j$ such that $\II_j=\left\{I\in \II~|~ w_I\in (\frac{1}{j+1},\frac{1}{j}]\right\}$ for $j=1,\ldots,k-1$ and $\II_k=\left\{I\in\II ~|~ w_I\in(0,\frac{1}{k}]\right\}$. Each sub-instance $\II_j$ is packed separately using Next-Fit. Given an instance $\II$ of static bin packing, the cost of the $k$-bounded space Harmonic algorithm is $\HH_k\cdot \OPTS(\II)+k$. $\HH_k$ is a monotonically decreasing number that approaches  $\HH_\infty\approx 1.691$. $\HH_k$ quickly becomes very close to this number, for example, $\HH_{12}\approx 1.692$. As shown by \cite{lee1985simple}, no constant bounded space algorithm can achieve an approximation ratio better than $\HH_\infty$. For our analysis we need the following  stronger guarantee for the performance of an online static bin packing algorithm, where $A(\II)$ denotes the cost of algorithm $A$ on instance $\II$.

\begin{defn}\label{def:decomposable}
An online (static) bin packing algorithm $A$ is $(c, \ell)$-decomposable if for every instance $\II = \bigcup_{j=1}^{n}\II_j$,
$\sum_{j=1}^{n} A(\II_j) \leq 
c\cdot \OPTS(\II)+ n \cdot \ell$.
\end{defn}

In particular, plugging $n=1$, an algorithm $A$ being $(c,\ell)$-decomposable implies that for any instance $\II$ it holds that $A(\II)\leq c\cdot \OPTS(\II)+\ell$. 

\begin{lem}\label{lem:static-decomposable}
The following algorithms are decomposable:
\begin{enumerate}
    \item Next-Fit is $(c,1)$-decomposable where $c=1$ in the uniform size case and $c=\min\left\{2,\frac{1}{1-\beta}\right\}$ in the non-uniform size case. 
    \item $k$-Harmonic is $(\HH_k,k)$-decomposable.
\end{enumerate}
\end{lem}

\begin{proof}

{\bf Proof of $(1)$:}
Next-Fit holds a single active bin that is still accepting intervals. The rest of the bins are full in the uniform size case and at least $\max\{\frac{1}{2},1-\beta\}$-full in the non-uniform size case. Similarly, in an instance decomposed into $n$ sub-instances there are $n$ active bins, while the rest of the bins are full in the uniform size case and at least $\max\{\frac{1}{2},1-\beta\}$-full in the non-uniform size case. This translates to a total cost of at most $\OPTS(\II)+n$ in the uniform size case and $\min\left\{2,\frac{1}{1-\beta}\right\}\cdot \OPTS(\II) + n$ in the non-uniform size case.

{\bf Proof of $(2)$:}
The $k$-bounded harmonic algorithm is composed of $k$ copies of the Next-Fit algorithm. The $k-1$ first copies (of the biggest items) can be seen as uniform size bin packing since exactly $j$ items are packed in each bin in the $j$-th copy. In the $k$-th copy sizes are not uniform, though for the sake of the analysis of the harmonic algorithm a bin is considered as full if it is $1-\frac{1}{k}$ full. Thus, this copy of Next-Fit is also $(1,1)$-decomposable. Decomposing each copy of Next-Fit leads to an additional cost of $n-1$, overall $k\cdot (n-1)$.
\end{proof}

\section{Improving the Approximation For Non-Uniform Sizes}\label{app:improve}

In this section we show how to use ideas from the analysis of the harmonic algorithm for the static bin packing to improve the performance of algorithms for the dynamic bin packing. This can be done for any algorithm for the dynamic case (with certain good properties). The reduction is given as Algorithm \ref{HarmonizeOffline2}.

\begin{algorithm}[h]
\SetAlgoLined
\DontPrintSemicolon
Let $A$ be an offline algorithm for the dynamic bin packing problem.\;
Partition $\II$ so that $\II = \bigcup_{j=1}^{k}I_j$, where $\II_j=\{I\in \II \ | \ w_I\in (\frac{1}{j+1}, \frac{1}{j}]\}$ for $j=1, \ldots, k-1$, and $\II_k= \{I\in \II \ | \ w_I\leq \frac{1}{k}\}$.\;
Schedule each subset $\II_j$ separately using $A$.\;
\caption{Partition Algorithm (parameter $k$)} \label{HarmonizeOffline2}
\end{algorithm}

\begin{lemma}\label{lem:harmonize-offline}
Let $A$ be an offline dynamic bin packing algorithm that for instance $\II$ with load vector $v$ when measured without the ceiling on each coordinate has a total cost of:
\begin{itemize}
    \item $c \cdot\|v\|_1 + f(v)$ for the uniform size case.
    \item $c_{\beta}\cdot\|v\|_1 + g(v)$ for the non-uniform size when parametrized by $\beta$,
\end{itemize}
where $f$ and $g$ are non-decreasing functions of the load vector. Then, for $k\geq 3$, the total cost of Algorithm \ref{HarmonizeOffline2} is at most $$\Pi_k\cdot\max\{c, c_{\frac{1}{k}}\cdot\frac{k-1}{k}\}\cdot OPT + (k-1)f(2v) +g(2v)$$
If $f$ is also concave then the total cost is at most:
$$\Pi_k\cdot\max\{c, c_{\frac{1}{k}}\cdot\frac{k-1}{k}\}\cdot OPT + (k-1) \cdot f\left(\frac{2v}{k-1}\right) +g(2v)$$

\end{lemma}

\begin{proof}
We define a new size for each interval. For $I\in \II_j, 1\leq j\leq k-1$ we set $w'_I = \frac{1}{j}$ and for $I\in \II_k$ we set $w'_I=w_I\cdot\frac{k}{k-1}$.

Let $v_j$ and $v'_j$ be the load vectors with respect to $w_I$ and $w'_I$ (respectively) of $\II_j, 1\leq j\leq k-1$. In any feasible schedule at most $j$ intervals can be scheduled on the same machine given both size functions. The interval sizes $w'$ in instance $\II_j$ are uniform thus, the total cost of scheduling $\II_j$ is at most $c\cdot\| v'_j\|_1 + f(v'_j)$.

The total cost of the schedule of $\II_k$ created by algorithm $A$ with respect to load vector $v_k$ is $c_{\beta}\cdot\|v_k\|_1+g(v_k), \beta\leq\frac{1}{k}$. The load of each machine with respect to $w'$ is larger as the size of each interval is multiplied by $\frac{k}{k-1}$. Thus, the total cost with respect to vector $v_k'$ is $c_{\frac{1}{k}}\cdot\frac{k-1}{k}\cdot\| v'_k\|_1+g(v'_k)$.

Summing over $\II_1,...,\II_k$, the total cost of the algorithm is at most
\begin{align*}
   &\max\{c,c_{\frac{1}{k}}\cdot\frac{k-1}{k}\}\cdot\sum_{j=1}^k\|v'_j\|_1 + \sum_{j=1}^{k-1}f(v'_j)+g(v'_k) \\ &\leq 
   \max\{c,c_{\frac{1}{k}}\cdot\frac{k-1}{k}\}\cdot\|v'\|_1 +(k-1)\cdot f(v')+g(v') 
\end{align*}

If $f$ is concave we can use Jensen's inequality to bound $\sum_{j=1}^{k-1}f(v'_j)$
from above by $(k-1)f(\frac{1}{k-1}\cdot \sum_{j=1}^{k-1}v'_j) \leq (k-1)f(\frac{v'}{k-1})$.

Any optimal solution must pay $1$ to pack $\HH_k$ of the load defined by $w'$. Thus, we can bound the optimal solution, $OPT\geq \frac{\|v'\|_1}{\HH_k}$. In addition, $v'\leq 2v$ and $\|v'\|_0=\|v\|_0$ which proves the lemma.
\end{proof}

\begin{corollary}
The total cost of Algorithm \ref{HarmonizeOffline2} is at most:
\begin{itemize}
    \item $2\cdot \left(1+\frac{1}{k-2}\right)\cdot \HH_k\cdot OPT + k\cdot \|v\|_0$ for $k\geq 4$ with Algorithm \ref{CoveringAlgorithm} as $A$.
    \item $\HH_k\cdot OPT+\left(\sqrt{k}+1\right)\cdot O\left(\sum_{t=1}^T\sqrt{v_t\log\mu}\right)$ with Algorithm \ref{DensityAlgorithm} as $A$.
\end{itemize}
\end{corollary}

\begin{proof}
As proven in Theorem \ref{thm:CoveringAlgorithm1} the cost of Algorithm \ref{CoveringAlgorithm} is at most $2\|\lceil v\rceil\|_1\leq 2\|v\|_1+\|v\|_0$ in the uniform size case and $\sum_{t=1}^T\lceil \frac{2v_t}{1-2\beta}\rceil \leq \frac{2}{1-2\beta}\|v\|_1 + \|v\|_0$ in the non-uniform size case. Thus, $c=2, c_{\frac{1}{k}}=\frac{2k}{k-2}$ and $f(v)=g(v)=\|v\|_0$. Thus, the total cost of Algorithm \ref{HarmonizeOffline2} with Algorithm \ref{CoveringAlgorithm} as $A$ is at most
$$ c_{\frac{1}{k}}\cdot \frac{k-1}{k}\cdot \HH_k\cdot OPT + k\cdot\|2v\|_0 \leq 2\cdot\left(1+\frac{1}{k-2}\right)\cdot\HH_k\cdot OPT+k\cdot\|v\|_0 $$

By Theorem \ref{thm:denseAlg} Algorithm \ref{DensityAlgorithm} has performance guarantee $c=1, c_{\frac{1}{k}}=\frac{k}{k-1}$ and $f(v)=g(v)=O\left(\sum_{t=1}^T\sqrt{v_t\log\mu}\right)$ which are concave functions. Thus, the total cost of Algorithm \ref{HarmonizeOffline2} with Algorithm \ref{DensityAlgorithm} as $A$ is at most
$$ \HH_k\cdot OPT+k\cdot O\left(\sum_{t=1}^T\sqrt{\frac{v_t}{k}\log\mu}\right) \leq \HH_k\cdot OPT + O\left(\sqrt{k\cdot OPT\cdot T\cdot\log\mu}\right)$$
where the inequality follows by Jensen's inequality.
\end{proof}

\section{Proofs Omitted}

\subsection{Proofs omitted from Section \ref{sec:covering}}\label{app:covering}

\begin{proof}[Proof of Lemma \ref{lem:cover1}]
We start with the non-uniform size case. Given a set of intervals $\II$ each with size less than $1/2$ (i.e., $\beta<\frac{1}{2}$), 
we show how to efficiently construct a $[\frac{1}{2}-\beta, 1]$-cover $C \subseteq \II$. Initially, we start with $C = \II$ and $v'$ is the load vector of the subset $C$.
Clearly, initially, for every $t$, $v'_t \geq \min\{v_t, \frac{1}{2}-\beta\}$. If for every $t$, $v'_t \leq 1$ then we are done. 
Otherwise, there exists a time $t$ such that $v'_t>1$. 
Consider the set of intervals $A\subseteq C$ that are active at time $t$. Using Lemma \ref{intersecting-intervals} with $\alpha=1$, we see that there exists an interval $I$ that observes load more than $1/2$ for its whole duration. Removing $I$, the load vector of the subset $A$ remains at least $\frac{1}{2}-\beta$. Since $A\subset C$ removing such an interval maintains the load at any time it intersects  at least $\frac{1}{2}-\beta$ also in the subset $C$. Thus, after iteratively removing these intervals we have that $v'_t \in[\min\{v_t,  \frac{1}{2}-\beta\}, 1]$.

The proof for the uniform case follows the same lines. Given a set of intervals $\II$, we initially set $C = \II$. Let $v'$ be the load vector of the subset $C$.
Clearly, initially, for every $t$, $v'_t \geq \min\{v_t, 1\}$. If for every $t$, $v'_t \leq 2$ then we are done. 
Otherwise, there exists time $t$ such that $v'_t>2$. Using Lemma \ref{intersecting-intervals} with $\alpha=2$, we see that there exists an interval $I$ that observes load strictly more than $1$ at any point. Hence, removing any such intervals (and using the fact that we are in the uniform case), the load vector of the subset $A$ remains at least $1$. Since $A\subseteq C$ this is true also for the subset $C$.
\end{proof}

\begin{proof}[Proof of Theorem \ref{thm:CoveringAlgorithm1}]
Consider an iteration $r$ of the while loop of Algorithm \ref{CoveringAlgorithm}. Let $\II^r$ denote the set of intervals in the beginning of iteration $r$, i.e., the intervals that have not been assigned in previous iterations. Denote by $v^r$ the load vector of $\II^r$, and by $C^r$ the cover that was obtained by Lemma \ref{lem:cover1} during iteration $r$. Let $v(C^r)$ be the load vector of $C^r$.

For the uniform case, since $v_t(C^r) \leq 2$ for all $t$ by construction, we have $\|v(C^r)\|_{\infty} \leq 2$, and the total cost of the algorithm in this iteration is $2 \cdot \|v^r\|_0$ by Lemma \ref{thm-ff}. By the properties of the cover, $|C^r(t)|\geq \min\{v^r_t, 1\}$ for every $t$. Hence, $\|v^r\|_1$ decreases by at least $\|v^r\|_0$.  Summing up over all iterations we get that the cost is at most $2\cdot \|v\|_1$. 

For the non-uniform case, in each iteration $r$, $\|v(C^r)\|_{\infty} \leq 1$, so First Fit schedules the corresponding intervals using one machine by Lemma \ref{thm-ff}. Hence, if the maximum size interval is at most $\frac{1}{4}$, by the properties of the cover, $|C^r(t)|\geq \min\{v^r_t, \frac{1}{2} - \beta \}$ for every $t$, so summing up over all iterations, the algorithm pays at most $\sum_{t} \lceil \frac{v_t}{\frac{1}{2}-\beta}\rceil$. If $\beta > \frac{1}{4}$, let $W_t$ be the number of intervals with width larger than $\frac{1}{4}$ that are active at time $t$. Since the algorithm opens a separate machine of unit size for each of them, it pays cost $W_t$ at each time $t$. Let $v^w$ and $v^n$ be the load vector of the intervals that have size more than $\frac{1}{4}$ and at most $\frac{1}{4}$ respectively, and let $\beta_n$ be the largest size of intervals in $v^n$. At each time $t$ the algorithm pays:
\begin{equation*}
   W_t + \lceil \frac{v_t^n}{\frac{1}{2}-\beta_n}\rceil  \leq W_t + \lceil 4 v_t^n \rceil = \lceil W_t + 4 v_t^n \rceil \leq 4 \cdot \lceil \frac{1}{4} \cdot W_t + v_t^n \rceil \leq 4 \cdot \lceil v_t^w + v_t^n \rceil = 4 \lceil v_t \rceil 
\end{equation*}

In the above, the first inequality follows from $\beta_n \leq \frac{1}{4}$ and the subsequent equality from $W_t$ being an integer. The second inequality is due to $\lceil \alpha x \rceil \leq \alpha \lceil x \rceil$  for $\alpha$ integer, and the final inequality is based on the fact that $v_t^w \geq \frac{1}{4} \cdot W_t$ since wide intervals have by definition size larger than $\frac{1}{4}$. Summing over all $t$, the total cost of the algorithm is at most: 
$$
\sum_t 4 \lceil v_t \rceil = 4 \|v\|_1
$$
\end{proof}

\subsection{Proofs omitted from Section \ref{sec:density-offline}}\label{app:density-offline}

\begin{proof}[Proof of Lemma \ref{lem:density}]
Let $t$ be a time with $ v_t\geq 2 + 4 \ln \mu$. We partition the intervals in $\II(t)$ into 1 + $\log_{1+\epsilon}\mu$ length classes $C_i$, with $\epsilon=\sqrt{\frac{D}{v_t}}$ where $D = 2 + 4\ln\mu$. The $i$th class contains all intervals in $\II(t)$ whose length is in the range  $[(1+\epsilon)^{i-1}, (1+\epsilon)^i)$. Note that since $v_t\geq 2 + 4 \ln\mu$, then $\epsilon \leq 1$.

Consider the intervals in the $i$th class $C_i$, and let $\ell=(1+\epsilon)^{i-1}$. By our partition, all lengths of intervals in $C_i$ are in the range $[\ell, \ell(1+\epsilon))$. Furthermore, as they all belong to $\II(t)$ (and are hence active at time $t$), the starting time of each interval $I\in C_i$ is in the range $(t-\ell(1+\epsilon), t]$. We next, further partition the intervals in $C_i$ to $(1+\frac{1}{\epsilon})$ sub-classes $C_{i,j}$ by their starting times. $C_{i,j}$ contains all intervals in $C_i$ whose starting time is in the time interval $(t-\ell(1+\epsilon)+ (j-1)\epsilon \ell,t-\ell(1+\epsilon)+ j \cdot \epsilon \ell]$. 

Overall, the total number of sets in our partition is at most, 
\begin{align*}
\left(\frac{1}{\epsilon}+1\right)\left( 1 + \frac{\ln\mu}{\ln(1+\epsilon)} \right) 
    & \leq \left(\frac{1}{\epsilon}+1\right) \left(1 + \frac{2\ln\mu}{\epsilon} \right) 
     \leq \frac{2 + 4\ln\mu}{\epsilon^2} = \frac{D}{\epsilon^2}
\end{align*}
where the first inequality follows since for $\epsilon \leq 1$, $\ln(1+\epsilon)\geq \frac{\epsilon}{2}$.

Hence, one of the sets must contain a load of at least $v_t \cdot \frac{\epsilon^2}{D} \geq 1$ at time $t$. This means that in the uniform case, where each interval has size $1/\uni$ for some integer $\uni$, at least one set contains at least $\uni$ intervals. In the non-uniform case it means there exists a set with size at least $1$. Given a max size of $\beta$, this set contains a subset of size at least $\max\{\frac{1}{2},1-\beta\}$. As all the intervals are of length $[\ell, \ell(1+\epsilon))$, and their starting time is at most $\ell\epsilon$ apart, it is possible to open a machine of length at most $\ell(1+2\epsilon) = \ell\left(1+ 2\sqrt{\frac{D}{v_t}}\right)$ for the selected subset of size at least $1/c$.
\end{proof}

\end{document}